\documentclass[english]{article}
\usepackage[T1]{fontenc}
\usepackage[latin9]{inputenc}
\usepackage{amsmath}
\usepackage{amsthm}
\usepackage{amssymb}

\makeatletter
\theoremstyle{definition}
\newtheorem{defn}{\protect\definitionname}
\theoremstyle{plain}
\newtheorem{prop}{\protect\propositionname}
\theoremstyle{plain}
\newtheorem{lem}{\protect\lemmaname}
\theoremstyle{plain}
\newtheorem{thm}{\protect\theoremname}
\theoremstyle{plain}
\newtheorem{cor}{\protect\corollaryname}

\makeatother

\usepackage{babel}
\providecommand{\corollaryname}{Corollary}
\providecommand{\definitionname}{Definition}
\providecommand{\lemmaname}{Lemma}
\providecommand{\propositionname}{Proposition}
\providecommand{\theoremname}{Theorem}

\begin{document}
\title{Putting the Agents Back in the Domain: A Two-Sorted Term-Modal Logic}
\author{Bachelor Thesis in Philosophy at university of Copenhagen Summer 2017}
\author{Supervisor: Vincent F. Hendricks \\ Written by: Andreas K. Achen}
\date{June 1st, 2017}
%\author{Written by: Andreas Achen}
\maketitle
\begin{abstract}
In the present paper, syntax and semantics will be presented for an expansion of ordinary n-agent QML with constant domain, non-rigid constants, rigid variables and including both functions, relations, and equality. Further, the number of agents will be specified axiomatically thus ensuring maximal flexibility wrt. the cardinality of the set of agents. Domain, variables, and constants will be partitioned in an agent-part and an object-part and the syntax will be expanded to include strings of the type $\forall xK_{x}\phi$ as wff's of the language. This will enhance expressiveness regarding the epistemic status of agents. Such a term-modal version of the logic $\boldsymbol{K}$ is shown to be sound and complete wrt. the class of (appropriate) frames, and a term-version of $\boldsymbol{S4}$ is shown to be sound and complete wrt. the class of (appropriate) frames in which the relations are transitive. It should be noted that completeness is shown via the framework of canonical models and thus allows for non-complicated generalizations to other logics than the term-versions of $\boldsymbol{K}$ and $\boldsymbol{S4}$. \thispagestyle{empty}\newpage{}
\end{abstract}
\[
\text{\emph{"Remember... All I'm offering is the truth. Nothing more."}}
\]

\[
-Morpheus
\]

\thispagestyle{empty}

\newpage{}

\tableofcontents{}\thispagestyle{empty}

\newpage{}

\section{Introduction \setcounter{page}{1}}

In a passage of \cite{Descartes} Descartes reflects on what can be
known and famously concludes that ``I think, therefore I am''. As
the insight presented in ``the cogito'' has a central place in classic
epistemology it is of natural interest for the student of formal epistemology
to answer the question of how such a ``quote'' might be formalized.
What is noticeable is the reflexive nature of the cogito, i.e. the
fact that it deals with an agent's reflections on itself. One might
go about formalizing the cogito in some epistemic modal logic: Propositional
modal logic, represented by the Backus-Naur form, is given as

\[
P::=\lnot P\mid P\rightarrow Q\mid K_{i}P
\]

\noindent and to express the sentiment of the cogito we would have
to device propositional formulas $P_{1},...,P_{n}$ and interpret
$P_{i}$ as ``agent $i$ is''. Under this ascription, $K_{i}P_{i}$
reads ``agent $i$ knows that she is'', and the question then becomes
how to formalize ``I think''. A natural choice would be $K_{i}\phi$
for some tautology $\phi$, since for instance any thinking being
ought to ``know'' e.g. $p\vee\neg p$, and anyone knowing such a
thing is certainly a thinking being. The cogito then, would amount
to \\
 $K_{i}(p\vee\neg p)\rightarrow K_{i}P_{i}$, the $K_{i}$'s being
modal operators, and the index set in one-to-one correspondence with
some fitting set of agents. Three things are noteworthy; first, in
order to express the cogito for a set of agents one would have to
do so rather crudely by conjunction and secondly, any relation between
``agent $i$'' and the subscript of the modal operator has to be
stated in meta-language, i.e. the agents are nowhere to be seen in
the semantics and thirdly, dealing with propositional logic there
is no intra-logical connection between the entities of which the propositions
speak and the semantics. The situation is slightly better in standard
first-order modal logic due to the enhanced expressivity that comes
with the introduction of quantifiers and relations; represented by
its Backus-Naur form standard first-order modal logic yields the following
syntax:

\[
\phi::=R(t_{1},...,t_{n})\mid\lnot\phi|\phi\rightarrow\psi\mid\forall x\phi|K_{i}\phi
\]
where $R$ is an $n$ \textendash{} ary relation symbol, and all the
$t_{i}$'s are terms\footnote{Will be defined later in the context of two-sorted term-modal logic.}.
In this language, the sentence ``agent $i$ knows that she is''
is formalized along the lines of $\exists xK_{i}(x=\alpha)$, assuming
that a constants $\alpha$ for each agent, and the cogito would amount
to something along the lines of $K_{i}(\phi\vee\neg\phi)\rightarrow\exists xK_{i}(x=\alpha)$.
What is important to notice is that in ordinary first-order $n$ \textendash{}
agent modal logic we have to state the relation between the subscript
$i$ of the modal operator and the referent of $\alpha$ in meta-language,
i.e. it is not built-in in the semantics when ``quantifying'' over
agents by means of the conjunction. What we really want is the ability
to state $\forall y\big(K_{y}(\phi\vee\neg\phi)\rightarrow\exists xK_{y}(x=y)\big)$
for some agent-designating variable $y$, i.e. we want modal operators
indexed by terms, we want the ability to distinguish between those
terms that are appropriate for modal indexation and those that are
not, and we want the agents back in the domain! In Hintikka's landmark
1962 book ``Knowledge and Belief: An Introduction to the Logic of
the Two Notions'' he lays the foundation of what is today called
\emph{epistemic logic. }In this context he discusses syntax for formulas
of the form $K_{\alpha}P$ for some individual constant $\alpha$,
and writes in a footnote:
\begin{quote}
``Strictly speaking, we ought to distinguish those free individual
symbols which can only take names of persons as their substitution-values
and which can therefore serve as subscripts of epistemic operators
from those which cannot do so.'' \cite[pp. 11]{Hintikka1962}
\end{quote}
Nevertheless, Hintikka does not present either syntax or semantics
for such a logic, which is what will be accomplished here.

\noindent The purpose of the present paper is thus to first state
syntax and semantics for a two-sorted term-modal logic with constant
domain, non-rigid constants, and equality and secondly to state and
prove soundness and completeness. As such the aim is simply for me
to hone my technical skills by treating a logic somewhat more complex
than what is normally encountered on undergraduate level. The following
is a merge between the exposition of a term-modal logic in \cite{Rendsvig-StuS2010}
and the exposition of a many-sorted modal logic in \cite{Rendsvig2011_MAthesis}.
To remain true to the epistemic interpretation and for notational
convenience I remain two-sorted even though the generalization to
general many-sorted logic is straightforward. The interested reader
is encouraged to do the proper adjustments in order to obtain the
many-sorted version of the logic presented here.

\noindent First, language and syntax for a language for a two-sorted
term-modal logic is stated in section $2$. Then, in section $3$
we present semantics including frames, interpretations, models and
valuations leading up to the notions of truth and validity in $3.1$
and $3.2$. In section $4$ we turn to an axiomatic system for the
logic, and in $5$ a couple of useful results are proved as preparation
for a proof of soundness in section $6$. Sections $7$ and $8$ contains
the bulk of the work in this paper, namely completeness results for
the two-sorted term-modal versions of $\boldsymbol{K}$ and \textbf{$\boldsymbol{S4}$},
and section $9$ relates the results to Hintikka's thoughts on the
axiom $\boldsymbol{4}$ as presented in \cite{Hintikka1962}. As the
completeness proof is lengthy and technical we present a quick overview
here meant to ease the acquisition.

\subsection{A Brief Survey of the Proof of Completeness}

The key insight regarding the overall proof strategy is contained
in Proposition \ref{Prop. IFF} which states the following: Showing
completeness essentially amounts to, given a consistent set of formulas,
constructing an appropriate model satisfying this set of formulas.
The next part of the completeness proof thus consists of constructing
this model starting with the worlds of the model, which we choose
to be a fitting maximally consistent set of formulas. Lindenbaum's
Lemma (Lemma \ref{Lemma Lindenbaum}) and the Saturation Lemma (Lemma
\ref{Lemma Saturation}) ensures that it is possible to extend a given
set of formulas to a maximally consistent set, and Lemma \ref{Lemma properties of MCS}
then states some crucial properties of such maximally consistent sets.
Definition \ref{Def. Canonical model} then presents \emph{the canonical
model, }and Propositions \ref{Prop. Well-defd domain}, \ref{Prop. Existence},
and \ref{Prop. Uniformity} are technical results ensuring that the
canonical model is well-defined and stands in appropriate relations
to the semantics of the logical connectives. Having defined the worlds
of the canonical model as sets of formulas the idea is that a world
should satisfy all and only those formulas which enjoy membership
in this particular world, and that the construction actually has this
feature is the content of the Truth Lemma (Lemma \ref{Lemma truth}).
With the Truth Lemma in hand the canonical class Theorem (Theorem
\ref{Theorem Canonical class}) is stated and proved, the consequence
of which is that any appropriate logic is complete wrt. the class
of canonical models. As an application of the above, completeness
of the ``term-modal'' version of $\boldsymbol{S4}$ is proven in
Theorem \ref{Theorem completeness S4}.

\section{The Language $\mathcal{L}_{TM}^{^{\overline{\sigma}}}$ and Syntax}

Fix an \emph{agent-set $\boldsymbol{\mathcal{A}}=\{\alpha_{1}...\alpha_{n}\}$
}at the outset. In this section we add these agents to the domain
of quantification resulting in a \emph{two-sorted term-modal logic}.
The two sorts will correspond to \emph{agents }and \emph{objects }respectively,
and ``term'' refers to the ability to quantify over indexes in modal
operators thus enhancing the ability to express agents reflecting
reflexively about themselves and other agents. In this section language
and syntax for a two-sorted term-modal logic with constant domain,
non-rigid constants, and equality is defined. We start with the syntax.
It is worth noting that the language is not parametrized by the agent-set
\textendash{} rather the number of agents is specified axiomatically\footnote{See axiom $\boldsymbol{N}$ in section 4.}.
Throughout, if $v=(x_{1},...,x_{n})$ is a vector, denote by $v_{i}$
the $i$'th element i.e $v_{i}=x_{i}$.
\begin{defn}
\label{Def. Language}{[}Language{]}. A \emph{two-sorted term-modal
language }$\mathcal{L}_{TM}^{^{\overline{\sigma}}}$ consists of
\end{defn}
\begin{enumerate}
\item A set $\overline{\sigma}:=\{agt,obj\}$ of two \emph{sorts.}
\item A countable, infinite set \emph{VAR }of variables, each assigned a
sort $\sigma\in\overline{\sigma}$. Let \emph{$VAR_{\sigma}$} denote
the set of $\sigma$ \textendash{} variables and assume that both
$VAR_{agt}$ and $VAR_{obj}$ are infinite.
\item A countable (possibly empty) set \emph{CON} of constants with each
constant assigned a sort. Denote the set of $\sigma$ \textendash{}
constants $CON_{\sigma}$.
\item A countable (possibly empty) set \emph{FUN }of function symbols, each
assigned an arity\emph{ $\alpha\in\overline{\sigma}^{k+1}$} for some
$k\geq0$. Denote the set of function symbols with $\alpha$ \textendash{}
arity $FUN_{\alpha}$.
\item A countable (possibly empty) set \emph{REL }of relation symbols, each
assigned an arity $\beta\in\overline{\sigma}^{k}$ for some $k\geq1$.
Denote the set of relations with $\beta$ \textendash{} arity $REL_{\beta}$.
\item The equality symbol $=$.
\item A set of modal operators $K_{t}$ indexed by \emph{agent-referring
terms} (to be defined shortly).
\item The logical connectives $\neg$ and $\rightarrow$.
\item The universal quantifier $\forall$.
\end{enumerate}
A couple of comments are in order before we proceed to define \emph{terms.
}First, the set $\overline{\sigma}$ of sorts is meant to model the
distinction between \emph{agents }and \emph{objects. }For readability
the sorts are simply denoted \emph{``agt'' }and \emph{``obj''
}but readers interested in generalizing to n-sorted logic might wish
to think of these in more generic terms\footnote{One option would be to simply define the set of sorts as $\overline{\sigma}=\{\sigma_{1},\sigma_{2}\}$
thus making the generalization completely straightforward.}.\emph{ }For the special notion of \emph{arity }employed in 4. notice
that the arity is now a vector which $i$'th entry is the \emph{sort
of the $i$'th input term }except from the last entry describing the
sort of the resulting term. The case with relations are completely
analogous so we turn towards the notion of a \emph{term. }
\begin{defn}
\label{Def. Terms}{[}Terms{]} The set $Term(\mathcal{L}_{TM}^{^{\overline{\sigma}}})$
of \emph{terms }is defined inductively by
\end{defn}
\begin{enumerate}
\item $VAR\cup CON\subseteq Term(\mathcal{L}_{TM}^{^{\overline{\sigma}}})$,
and
\item If $f\in FUN_{\alpha}$ with $\alpha\in\overline{\sigma}^{k+1}$ and
$t_{1},...,t_{k}\in Term(\mathcal{L}_{TM}^{^{\overline{\sigma}}})$
where $t_{i}$ is of sort $\alpha_{i}$ then $f(t_{1},...,t_{k})\in Term(\mathcal{L}_{TM}^{^{\overline{\sigma}}})$.
\end{enumerate}
The set $Term_{\sigma}$ of $\sigma$ \textendash{} terms is defined
as the \emph{smallest} subset of $Term(\mathcal{L}_{TM}^{^{\overline{\sigma}}})$
such that
\begin{enumerate}
\item $VAR_{\sigma}\cup CON_{\sigma}\subseteq Term_{\sigma}$, and
\item If $f\in FUN_{\alpha}$ with $\alpha\in\overline{\sigma}^{k+1}$ such
that $\alpha_{k+1}=\sigma$ and $t_{1},...,t_{k}\in Term(\mathcal{L}_{TM}^{^{\overline{\sigma}}})$
where $t_{i}$ is of sort $\alpha_{i}$ then $f(t_{1},...,t_{k})\in Term_{\sigma}$.
\end{enumerate}
Specifically, the set $Term_{agt}$ is what we will call \emph{the
set of agent-referring terms.}

We can now qualify point 7. in\textbf{ }Definition \ref{Def. Language}
by defining the set $K$ of modal operators by

\[
K:=\{K_{t}|t\in Term_{agt}\}
\]

The set of $\mathcal{L}_{TM}^{^{\overline{\sigma}}}$ \emph{well-formed
formulas }is now definable:
\begin{defn}
\label{Def. wff}{[}Well-formed formulas{]} The set of $\mathcal{L}_{TM}^{^{\overline{\sigma}}}$
well-formed formulas is defined inductively by:
\end{defn}
\begin{enumerate}
\item If $t_{1},t_{2}\in Term(\mathcal{L}_{TM}^{^{\overline{\sigma}}})$
then the expression $t_{1}=t_{2}$ is an \emph{atomic }wff of the
language, and
\item If $P\in REL_{\beta}$ with $\beta\in\overline{\sigma}^{k}$ and $t_{1},...,t_{n}\in Term(\mathcal{L}_{TM}^{^{\overline{\sigma}}})$
where $t_{i}$ is of sort $\beta_{i}$ then $P(t_{1},...,t_{k})$
is an \emph{atomic }wff of the language, and
\item If $\phi$ and $\psi$ are \emph{atomic wffs}, $t\in Term_{agt}$,
and $x\in VAR$\emph{ }then the remaining wffs of the language is
defined by the Backus-Naur form:
\end{enumerate}
\[
\varphi::=\lnot\varphi\mid\varphi\rightarrow\psi\mid\forall x\varphi\mid K_{t}\varphi
\]

\noindent Please note that the remaining logical connectives are defined
in the ordinary way while the existential quantifier $\exists$ and
the modal operators $P_{t}$ are defined as duals of $\forall$ and
$K_{t}$ respectively \textendash{} this is also completely standard.

\medskip{}

\noindent Before moving on we define the notion of \emph{a free variable}.
For the most part it is completely standard but for the set $K$ of
modal operators something new is going on.
\begin{defn}
{[}Free variable{]}\label{Def. Free variable} We define \emph{a free
variable }inductively by
\end{defn}
\begin{enumerate}
\item For any \emph{term} $t$ the free variables are just the variables
occurring in $t$, and
\item For any \emph{atomic formula }$\phi$ the free variables are the free
variables of the terms of $\phi$, and
\item For formulas $\phi,\psi$ the free variables of $\neg\phi$ and $\phi\rightarrow\psi$
are just the free variables of $\phi$ and $\psi$, and
\item For each formula of the form $K_{t}\phi$ the free variables are just
the variables of $t$ joined with the free variables of $\phi$, and
\item The free variables of $\forall x\phi$ are any free variable of $\phi$
excluding $x$.
\end{enumerate}
Any occurring variable \emph{not free }is referred to \emph{as bound.
}Furthermore, any formula \emph{without any free variables }is called
a \emph{sentence. }Where $\phi$ is a formula, $t$ is a term, and
$x$ a variable we denote by $\phi(t/x)$ the result of substituting
any \emph{free occurrence of $x$ }for $t$, demanding that \emph{no
free variable of $t$ becomes bound by doing so}\footnote{What then do we take $\phi(t/x)$ to mean, if in fact some free variable
in $t$ \emph{does }become bound by the substitution? We simply substitute
$\phi$ for an appropriate \emph{alphabetic variant }which is always
possible. See \cite{Hughes_&_Cresswell_New.Intro} pp. 241.}.

\medskip{}

\noindent This concludes the syntactic definitions, which do not differ
much from the syntax of standard many-sorted modal logic as presented
in e.g. in \cite{Rendsvig2011_MAthesis}. The new part is that we
can quantify over indexes for modal operators, that is expressions
of the type $\forall xK_{x}\phi$ are now wffs of the language. Of
course such syntactic add-ons are in need of a semantic counterpart;
thus turning to semantics we put the agents $\boldsymbol{\mathcal{A}}$
back in the domain!

\section{Semantics}

Before moving on one bit of notation is needed. Where $A$ and $B$
are sets, take $A\dot{\cup}B$ to mean \emph{the disjoint union of
$A$ and $B$.}
\begin{defn}
\label{Def. Frame}{[}$TM_{\boldsymbol{\mathcal{A}}}^{\overline{\sigma}}$
\textendash{} frame{]} A $TM_{\boldsymbol{\mathcal{A}}}^{\overline{\sigma}}$
\textendash{} frame $\mathcal{F}$ for the language $\mathcal{L}_{TM}^{^{\overline{\sigma}}}$
is a triple $\mathcal{F}:=\langle\mathcal{W},\mathcal{R},DOM\rangle$
where
\end{defn}
\begin{enumerate}
\item $\mathcal{W}$ is a non-empty set of \emph{worlds, }and
\item $\mathcal{R}$ is a map associating binary \emph{accessibility relations
}on $\mathcal{W}$ to each agent, that is $\mathcal{R}:\boldsymbol{\mathcal{A}}\rightarrow\mathcal{P}(\mathcal{W}\times\mathcal{W})$,
and
\item $DOM=DOM_{agt}\dot{\cup}DOM_{obj}$ is the domain of quantification,
where neither of the cojoints are empty and $DOM_{agt}=\boldsymbol{\mathcal{A}}.$
\end{enumerate}
Take $R_{i}$ to mean $\mathcal{R}(\alpha_{i})$ and write $wR_{i}w'$
whenever $w,w'\in W$ are such that $(w,w')\in\mathcal{R}(\alpha_{i})$.

\medskip{}

\noindent Next up is the notion of an \emph{interpretation }but before
bringing the technical definition a few comments are in order. We
are currently defining a logic with non-rigid functions, relations,
and constants which is naturally reflected in the definition above.
Also, as we now have to take care what sort of term goes into each
place of functions and relations the notion of \emph{arity }is slighty
more complex than the ordinary natural number cf. Definition \ref{Def. Language}.
\begin{defn}
\label{Def. Interpretation}{[}Interpretation{]} An \emph{interpretation
$\mathcal{I}$ }is a map such that simultaneously
\end{defn}
\begin{enumerate}
\item For each $\beta\in\overline{\sigma}^{k}$ let $\mathcal{I}:REL_{\beta}\times\mathcal{W}\rightarrow\mathcal{P}(\prod_{i\in\{1...k\}}DOM_{\beta_{i}})$,
and
\item For each $\alpha\in\overline{\sigma}^{k+1}$ let $\mathcal{I}:FUN_{\alpha}\times\mathcal{W}\rightarrow\mathcal{P}(\prod_{i\in\{1...k+1\}}DOM_{\alpha_{i}})$\footnote{In reality, $\mathcal{I}$ maps into the power set of a subset of
$\prod_{i\in\{1...k+1\}}DOM_{\alpha_{i}}$: Namely the subset defined
by the rule that if $\{DOM_{\alpha_{i}}\}_{i\in{\{1...k+1\}}}$ and
$\{DOM_{\gamma_{i}}\}_{i\in{\{1...k+1\}}}$ are such that $DOM_{\alpha_{i}}=DOM_{\gamma_{i}}$
for $i\in{\{1...k\}}$ it hold that $DOM_{\gamma_{k+1}}=DOM_{\alpha_{k+1}}$. }, and
\item $\mathcal{I}(=,w)=\{(d,d)|d\in DOM\}$ for every $w\in\mathcal{W}$,
and
\item For each $\sigma\in\overline{\sigma}$ let $\mathcal{I}:CON_{\sigma}\times\mathcal{W}\rightarrow DOM_{\sigma}$.
\end{enumerate}
\begin{defn}
{[}Model{]}\label{Def. Model} The two-tuple $\langle\mathcal{F},\mathcal{I}\rangle$
consisting of a \emph{frame }and an \emph{interpretation }is called
\emph{a model} and we write $\mathcal{M}=\langle\mathcal{F},\mathcal{I}\rangle.$
In this case we say that $\mathcal{M}$ is \emph{based on} $\mathcal{F}$.
\end{defn}
Having dealt with relation symbols, function symbols, equality and
constants the only remaining task is to interpret free variables.
For this we need a \emph{valuation:}
\begin{defn}
{[}Valuation{]}\label{Def. Valuation} A \emph{valuation $v$ }is
a (surjective) map such that for each $\sigma\in\overline{\sigma}$
we have $v:VAR_{\sigma}\rightarrow DOM_{\sigma}$.
\end{defn}
For technical reasons we need one more definition before we turn to
\emph{truth. }Namely, we need the notion of an \emph{x-variant }to
deal with formulas involving $\forall$.
\begin{defn}
\label{Def. x-variant}{[}$x$ - variant{]} An \emph{x-variant $v'$
}of a valuation $v$ is a valuation such that $v'$ and $v$ agrees
on all variables \emph{except possibly for $x$.} We note that any
valuation is an $x$ \textendash{} variant of itself.
\end{defn}
For brevity we shall write $t^{w,v}$ for the extension of the term
$t$ at world $w$ under valuation $v$, that is:

\[
t^{w,v}=\begin{cases}
v(t) & t\in VAR\\
\mathcal{I}(t,w) & t\in CON
\end{cases}
\]

It is useful to note that for any \emph{agent-referring term $t$}
the composite $\mathcal{R}(t^{w,v})$ denotes an accessibility relation
in $\mathcal{R}$. 

\subsection{Truth}
\begin{defn}
\label{Def. non-modal truth}{[}Truth for non-modal formulas{]} Let
$\mathcal{M}=\langle\mathcal{W},\mathcal{R},DOM,\mathcal{I}\rangle$
be a model, $w\in\mathcal{W}$, and $v$ a valuation. We define \emph{truth
in the model $\mathcal{M}$ at $w\in\mathcal{W}$} for the formula
$\phi\in\mathcal{L}_{TM}^{^{\overline{\sigma}}}$ inductively by:
\end{defn}
\begin{enumerate}
\item $\mathcal{M},w\models_{v}P(t_{1}...t_{k})$ iff $(t_{1}^{^{w,v}}...t_{k}^{w,v})\in\mathcal{I}(P,w)$,
and
\item $\mathcal{M},w\models_{v}t_{1}=t_{2}$ iff $t_{1}^{w,v}=t_{2}^{w,v}$,
and
\item $\mathcal{M},w\models_{v}\neg\phi$ iff it is not the case that $\mathcal{M},w\models_{v}\phi$,
and
\item $\mathcal{M},w\models_{v}\phi\rightarrow\psi$ iff \emph{either} $\mathcal{M},w\models_{v}\psi$
\emph{or $\mathcal{M},w\models_{v}\neg\phi$, }and
\item $\mathcal{M},w\models_{v}\forall x\phi$ iff $\mathcal{M},w\models_{v}\phi$
\emph{for every }$x$ \emph{- variant $v'$ of $v$.}
\end{enumerate}
And we notice that whenever $\phi$ is a \emph{sentence }we can omit
explicit mentioning of the valuation and simply write $\mathcal{M},w\models\phi$.
\begin{defn}
\label{Def. modal truth}{[}Truth for modal formulas{]} Let $\mathcal{M}=\langle\mathcal{M},\mathcal{R},DOM,\mathcal{I}\rangle$
be a model, $w\in\mathcal{W}$, $v$ a valuation, $\phi\in\mathcal{L}_{TM}^{^{n,\overline{\sigma}}}$,
and $t$ an \emph{agent-referring term}. We define \emph{truth in
the model $\mathcal{M}$ at $w$} for modal formulas inductively by:
\end{defn}
\begin{enumerate}
\item $\mathcal{M},w\models_{v}K_{t}\phi$ iff $\mathcal{M},w'\models_{v}\phi$
for every $w'$ such that $(w,w')\in\mathcal{R}(t^{w,v})$, and
\item $\mathcal{M},w\models_{v}P_{t}\phi$ iff $\mathcal{M},w'\models_{v}\phi$
for at least one $w'$ such that\\
 $(w,w')\in\mathcal{R}(t^{w,v})$.
\end{enumerate}

\subsection{Validity}

Above we defined the most basic notion of \emph{truth} namely that
of \emph{truth in a world of a model under a given valuation}. As
the aim of the project is to state and prove \emph{completeness,}
it shall be of paramount importance for us to be able to distinguish
truth at different levels. As we shall see later the notion of \emph{satisfiability
}will play a key role in proving completeness.
\begin{defn}
\label{def. Validity}{[}Validity and Satisfiability{]} A formula
$\phi$ is \emph{satisfiable }if there exists a model $\mathcal{M}$,
a world $w$ and a valuation $v$ such that $\mathcal{M},w\models_{v}\phi$.
Further, if $\phi$ is such that $\mathcal{M},w\models_{v}\phi$ for
\emph{all valuations }we say that $\phi$ \emph{is valid at world
$w$, }and we write $\mathcal{M},w\models\phi$. We say that a formula
$\phi$ is \emph{valid in the model $\mathcal{M}$ if $\mathcal{M},w\models\phi$
for all $w\in\mathcal{W}$ }and we write $\mathcal{M}\models\phi$.
If a formula $\phi$ is valid in \emph{all models based on a frame
}$\mathcal{F}$ we say that $\phi$ is valid in $\mathcal{F}$ and
we write $\mathcal{F}\models\phi$. We say that a formula $\phi$
is \emph{valid on the class $F$ of frames if $\mathcal{F}\models\phi$
for all $\mathcal{F}\in F$ }and we write this $F\models\phi$. If
$\phi$ is valid \emph{on the class of all frames }we simply say that
$\phi$ is \emph{valid }and write $\models\phi$.
\end{defn}
Please recall that we require formulas to be \emph{finite. }Yet, sometimes
we wish to speak about infinite strings of symbols so it is advantageous
for us to introduce a bit of notation to deal with this. So, whenever
$\Gamma$ is an \emph{arbitrarily large }set of formulas we write
$\mathcal{M},w\models_{v}\Gamma$, $\mathcal{M},w\models\Gamma$ and
so forth to mean the obvious. We ask the reader to also note that
whenever $S$ is a class of \emph{models a model from $S$ simply
means some model $\mathcal{M}$ for which $\mathcal{M}\in S$, }while
if $F$ is a class \emph{of frames a model from $F$ is a model based
on some frame $\mathcal{F}\in F$, }that is $\mathcal{M}=\langle\mathcal{F},\mathcal{I}\rangle$\emph{
}for some interpretation $\mathcal{I}$\emph{. }I shall write $\mathcal{M}\in F$
for a model based on some frame $\mathcal{F}\in F$, and write that
$\mathcal{M}$ is a model \emph{from $F$. }We are now in position
to define \emph{the semantic consequence relation:}
\begin{defn}
\label{Def. Semantic Consequence}{[}Semantic Consequence{]} Let $\phi$
be a formula, $\Gamma$ be a set of formulas, and $S$ a \emph{class
of structures }(either models or frames)\emph{. }We say that $\phi$
is a \emph{semantic consequence of $\Gamma$ }and write $\Gamma\models_{S}\phi$
if it holds \emph{in all models $\mathcal{M}$ from $S$, for all
valuations $v$ and all worlds $w$ of $\mathcal{M}$ that if $\mathcal{M},w\models_{v}\Gamma$
then $\mathcal{M},w\models_{v}\phi$.}
\end{defn}
Having defined the semantic consequence relation we are in position
to clarify the notion of \emph{a semantically specified logic:}
\begin{defn}
\label{Def. Semantically Specified Logic}{[}Semantically Specified
Logic{]} Given a language $\mathcal{L}$, and a class $S$ of \emph{structures}
(either models or frames) formulated in $\mathcal{L}$, we define
the logic $L_{S}:=\{\phi\:|\:\models_{S}\phi\}$. 
\end{defn}
We proceed directly to the first Proposition of the project.
\begin{prop}
\emph{\label{Prop. Principle of replacement}{[}Principle of Replacement{]}
Let $\phi$ be a formula, $x,y$ variables, $\mathcal{M}$ a model,
$w$ a world, and $v$ a valuation. Then it holds that if $v'$ is
an $x$ - variant of $v$ with $v(x)=v'(y)$ then }
\end{prop}
\[
\mathcal{M},w\models_{v}\phi\:\text{iff}\:\mathcal{M},w\models_{v'}\phi(y/x)
\]

\begin{proof}
Seeing that the only leeway after fixing model and world is the interpretation
of free variables via the valuation, and $\phi(y/x)$ differs only
from $\phi$ in that $\phi(y/x)$ has a free occurrence of $y$ everywhere
that $\phi$ has a free occurrence of $x$, and $v$ and $v'$ agrees
everywhere except for $x$ where $v(x)=v'(y)$ the Proposition follows.
\end{proof}
Having established a thorough \emph{semantic description }of the two-sorted
term-modal logic we now aim to provide the proof-theoretical counterpart;
indeed the overall purpose of the project is to show that these two
descriptions amount to the same thing. As such we now turn to the\emph{
axioms.}

\section{An Axiomatic System of $\boldsymbol{K_{TM}}^{\boldsymbol{\mathcal{A}},\overline{\sigma}}$}

In this section we define the two-sorted term-modal logic version
of the concept of \emph{normal logics. }We will denote the resulting
logic for n agents over the language\emph{ $\mathcal{L}_{TM}^{^{\overline{\sigma}}}$
}by $\boldsymbol{K_{TM}}^{\boldsymbol{\mathcal{A}},\overline{\sigma}}$
and begin by stating what \emph{axioms }are to hold. Then we proceed
to the definition of a $\boldsymbol{K_{TM}}^{\boldsymbol{\mathcal{A}},\overline{\sigma}}$-
proof before we are ready to define \emph{normality }in the new setting.
The section concludes with three Propositions that shall be of help
to us when proving completeness later on.

\subsection{Axiom Schemas of $\boldsymbol{K_{TM}}^{\boldsymbol{\mathcal{A}},\overline{\sigma}}$}

First of all, we let all \emph{substitution instances} of valid formulas
from Propositional modal logics be axioms of $\boldsymbol{K_{TM}}^{\boldsymbol{\mathcal{A}},\overline{\sigma}}$\footnote{That is, any formula obtained from a Propositional axiom by substituting
Propositional variables for formulas from \emph{$\mathcal{L}_{TM}^{^{n,\overline{\sigma}}}$.}}. Furthermore, we add the following axiom schemas:

Let $\phi\in\mathcal{L}_{TM}^{\overline{\sigma}}$ be any formula
and $y$ any variable \emph{free in $\phi$. }Then every instance
of

\begin{equation}
\forall x\varphi\rightarrow\varphi\left(y/x\right)\tag{\ensuremath{\boldsymbol{\forall}}}
\end{equation}

is an axiom. Further, for any term $t$

\begin{equation}
t=t\tag{\ensuremath{\text{\ensuremath{\boldsymbol{Id}}}}}
\end{equation}

are axioms. Also, in order to accommodate the partition of terms in
sorts we include for each $x\in VAR_{agt}$ and $y\in VAR_{obj}$
the axiom

\begin{equation}
x\neq y\tag{\ensuremath{\text{\ensuremath{\boldsymbol{MSD}}}}}
\end{equation}

Dealing with a language including equality we also add \emph{the Principle
of Substitutivity }such that for any variables $x,y$ and any formula
$\phi$ we have

\begin{equation}
(x=y)\rightarrow\big(\varphi(x)\rightarrow\phi(y)\big)\tag{\text{\ensuremath{\boldsymbol{PS}}}}
\end{equation}

as an axiom\footnote{Where $\phi$ involves the modal operator $K_{x}$ we obviously requires
$x$ and $y$ to be agent-referring.}. We note that as variables are rigid the stated version of $\text{PS}$
should be of no concern. We also add \emph{Existence of Identicals,
}such that for any constant $c$ 

\begin{equation}
(c=c)\rightarrow\exists x(x=c)\tag{\ensuremath{\boldsymbol{\exists}\text{\ensuremath{\boldsymbol{Id}}}}}
\end{equation}

is an axiom. For technical reasons we also need to make sure that
the agent-part of the domain always contain the appropriate number
of element. So, for any $x_{1},...x_{|\boldsymbol{\mathcal{A}}|},y\in VAR_{agt}$
we have as an axiom:

\begin{equation}
\exists x_{1}...x_{|\boldsymbol{\mathcal{A}}|}\big(x_{1}\neq x_{2}\wedge...\wedge x_{|\boldsymbol{\mathcal{A}}|-1}\neq x_{|\boldsymbol{\mathcal{A}}|}\wedge\forall y(y=x_{1}\vee...\vee y=x_{|\boldsymbol{\mathcal{A}}|})\big)\tag{\ensuremath{\text{\ensuremath{\boldsymbol{N}}}}}
\end{equation}

We further include a version of the axiom $\boldsymbol{K}$, such
that for any agent-referring term $t$ and any formulas $\phi$ and
$\psi$ we have

\begin{equation}
K_{t}(\phi\rightarrow\psi)\rightarrow(K_{t}\phi\rightarrow K_{t}\psi)\tag{\ensuremath{\text{\ensuremath{\boldsymbol{K}}}}}
\end{equation}

as an axiom of the system $\boldsymbol{K_{TM}}^{\boldsymbol{\mathcal{A}},\overline{\sigma}}$
too. We also choose to add the \emph{Barcan Formula }for interplay
between quantifiers and modal operators such that for any agent-referring
term $t$ and any variable $x\neq t$ we have the axiom

\begin{equation}
\forall xK_{t}\phi\rightarrow K_{t}\forall x\phi\tag{\ensuremath{\text{\ensuremath{\boldsymbol{BF}}}}}
\end{equation}

and lastly; for any variables $x$ and $y$, and any agent-referring
term $t$ we add the axiom \emph{Knowledge of Non-identity:}

\begin{equation}
(x\neq y)\rightarrow K_{t}(x\neq y)\tag{\ensuremath{\text{\ensuremath{\boldsymbol{KNI}}}}}
\end{equation}

And we note that the axiom \emph{Dual }is unnecessary since we have
simply defined the operator $P_{t}$ such that $P_{t}:=\neg K_{t}\neg$.

\subsection{Inference Rules of $\boldsymbol{K_{TM}}^{\boldsymbol{\mathcal{A}},\overline{\sigma}}$}

The strategy when defining a logic in a syntactical matter is to first
state a set $\Lambda$ of axioms and then to close it off under \emph{rules
of inference} such that the resulting logic can be defined as $\overline{\Lambda}$\footnote{The overline denoting \emph{closure.}}
by a slight abuse of notation\emph{.} First, we choose to include
\emph{modus ponens. }That is for any formulas $\phi$ and $\psi$

\noindent 
\begin{equation}
\frac{\varphi,\varphi\rightarrow\psi}{\psi}\tag{\ensuremath{\boldsymbol{MP}}}
\end{equation}

is a valid inference. Also, we include \emph{Knowledge Generalization
}among our rules of inference such that for every agent-referring
term $t$ we have

\noindent 
\begin{equation}
\frac{\phi}{K_{t}\phi}\tag{\ensuremath{\boldsymbol{KG}}}
\end{equation}

as a valid inference. Finally, if $\phi$ is without free occurrences
of $x$ we let

\noindent 
\begin{equation}
\frac{\phi\rightarrow\psi}{\phi\rightarrow\forall x\psi}.\tag{\ensuremath{\boldsymbol{Gen}}}
\end{equation}

\noindent Having established axioms and rules of inference we are
ready to define the notion of a $\boldsymbol{K_{TM}}^{\boldsymbol{\mathcal{A}},\overline{\sigma}}$\textendash{}
proof.
\begin{defn}
\label{Def. Proof}{[}$\boldsymbol{K_{TM}}^{\boldsymbol{\mathcal{A}},\overline{\sigma}}$
\textendash{} proof{]} A $\boldsymbol{K_{TM}}^{\boldsymbol{\mathcal{A}},\overline{\sigma}}$
\textendash{} proof is a finite sequence of formulas from \emph{$\mathcal{L}_{TM}^{^{\overline{\sigma}}}$,
}each of which is either a $\boldsymbol{K_{TM}}^{\boldsymbol{\mathcal{A}},\overline{\sigma}}$
\textendash{} axiom or the result of using some rule of inference
on one or more earlier formulas of the sequence. If $\phi$ is a formula
we say that $\phi$ is $\boldsymbol{K_{TM}}^{\boldsymbol{\mathcal{A}},\overline{\sigma}}$
\textendash{} provable if it is the last element in such a sequence
and we write $\vdash_{\boldsymbol{K_{TM}}^{\boldsymbol{\mathcal{A}},\overline{\sigma}}}\phi$. 
\end{defn}
Note that at this point we know nothing about the relation between
truth as defined in definitions \ref{Def. non-modal truth} and \ref{Def. modal truth}
on the one side and provability as just defined on the other. It will
be the aim of this project to establish the appropriate equivalence!
Three more concepts shall be defined here before turning to a couple
of handy results \textendash{} which will also yield an opportunity
to see $\boldsymbol{K_{TM}}^{\boldsymbol{\mathcal{A}},\overline{\sigma}}$
\textendash{} proofs in action.
\begin{defn}
\label{Def. Normality}{[}Normality{]} A \emph{normal} $n$ \textendash{}
agent two-sorted term-modal logic $\Lambda$ is any set of \emph{$\mathcal{L}_{TM}^{^{\overline{\sigma}}}$
\textendash{} }formulas containing all of the above axioms closed
under all of the above rules of inference.
\end{defn}
Whenever $\phi\in\Lambda$ we say that $\phi$ \emph{is a Theorem
of $\Lambda$, }and write $\vdash_{\Lambda}\phi$.\emph{ }If $\Lambda_{1}$
and $\Lambda_{2}$ are two logics such that $\Lambda_{1}\subseteq\Lambda_{2}$
we say that $\Lambda_{2}$ \emph{is} \emph{an extension of $\Lambda_{1}$. }
\begin{defn}
\label{Def. Deducibility}{[}$\Lambda$ \textendash{} deducibility{]}
Let $\Gamma$ be any set of wffs of the language, $\phi$ be any formula,
and $\Lambda$ a logic. Then $\phi$ \emph{is $\Gamma$ \textendash{}
deducible from $\Lambda$ }if there exist a finite subset $\Gamma_{0}\subseteq\Gamma$\emph{
}such that 

\[
\vdash_{\Lambda}\Gamma_{0}\rightarrow\phi.
\]

Whenever this obtains, write $\Gamma\vdash_{\Lambda}\phi$. Write
$\Gamma\not\vdash_{\Lambda}\phi$ if $\phi$ is \emph{not }$\Lambda$
\textendash{} deducible from $\Gamma$. 
\end{defn}

\begin{defn}
\label{Def. Consistency}{[}$\Lambda$ \textendash{} consistency{]}
Let $\Gamma$ be any set of wffs of the language. We say that $\Gamma$
is $\Lambda$ \textendash{} consistent if \emph{for all formulas $\phi$
we have}

\[
\Gamma\not\vdash_{\Lambda}\phi\wedge\neg\phi
\]

and $\Lambda$ \textendash{} inconsistent otherwise. Please note that
if some arbitrary set of wffs $\Gamma$ is inconsistent then there
has to exist $\Gamma_{0}\subseteq\Gamma$ finite which is incomplete
since proofs are finite by definition. 
\end{defn}
Having defined the framework we end this section by the proof-theoretical
counterpart to definition \ref{Def. Semantically Specified Logic}:
\begin{defn}
\label{Def. Proof Theoretically Specified Logic}{[}Proof-Theoretically
Specified Logic{]} If $\mathcal{L}$ is a language, and $K$ denote
a set of axioms together with some rules of inference formulated in
$\mathcal{L}$, we define the logic $L_{K}:=\{\phi\:|\:\vdash_{K}\phi\}.$ 
\end{defn}
Now, we can express completeness as the inclusion $L_{S}\subseteq L_{K}$
and soundness as $L_{K}\subseteq L_{S}$. We proceed to state and
prove some Lemmas that shall be of great use during the proof of completeness.

\section{A Couple of Handy Lemmas}
\begin{lem}
\emph{\label{Lemma 1}Let $x,y\in VAR$ and $t\in Term_{agt}$. Then}

\[
\vdash_{\boldsymbol{K_{TM}}^{\boldsymbol{\mathcal{A}},\overline{\sigma}}}(x=y)\rightarrow K_{t}(x=y).
\]
\end{lem}
\begin{proof}
Let $\phi(z)=K_{t}(x=z)$ for any $t\in Term_{agt}.$ Then we have
by Propositional calculus that $\vdash_{\boldsymbol{K_{TM}}^{\boldsymbol{\mathcal{A}},\overline{\sigma}}}(x=y)\rightarrow\big(K_{t}(x=x)\rightarrow K_{t}(x=y)\big)$.
Using the Theorem from Propositional calculus that

\[
\big(\phi\rightarrow(\psi\rightarrow\varphi)\big)\rightarrow\big(\varphi\rightarrow(\phi\rightarrow\psi)\big)
\]

yields $\vdash_{\boldsymbol{K_{TM}}^{\boldsymbol{\mathcal{A}},\overline{\sigma}}}K_{t}(x=x)\rightarrow\big((x=y)\rightarrow(K_{t}(x=y)\big)$.
It follows directly from \textbf{KG} and Id that $\vdash_{\boldsymbol{K_{TM}}^{\boldsymbol{\mathcal{A}},\overline{\sigma}}}K_{t}(x=x)$,
and so an application of \textbf{MP} yields the desired result.
\end{proof}
In the following the formula 

\[
(x=y)\rightarrow K_{t}(x=y)\tag{\ensuremath{\text{\ensuremath{\boldsymbol{KI}}}}}
\]

will be called \emph{Knowledge of Identity. }
\begin{lem}
\emph{\label{Lemma K-Dist}Let $t\in Term_{agt}$, and $\phi_{1},...,\phi_{n}$
be wffs of the language. Then the following holds}
\end{lem}
\[
\vdash_{\boldsymbol{K_{TM}}^{\boldsymbol{\mathcal{A}},\overline{\sigma}}}K_{t}\big(\phi_{1}\wedge...\wedge\phi_{n}\big)\rightarrow\big(K_{t}\phi_{1}\wedge...\wedge K_{t}\phi_{n}\big)
\]

In the following the formula

\[
K_{t}\big(\phi_{1}\wedge...\wedge\phi_{n}\big)\rightarrow\big(K_{t}\phi_{1}\wedge...\wedge K_{t}\phi_{n}\big)\tag{\ensuremath{\boldsymbol{KD}}}
\]

will be called \emph{K-Distribution.}
\begin{proof}
We omit the proof as it is fairly standard. See \cite{Hughes_&_Cresswell_New.Intro}
pp. 28.
\end{proof}
\begin{lem}
\emph{\label{Derived rule}Let $\phi$ and $\psi$ be wffs of the
language, and $t\in Term_{agt}$. Then it holds that}
\end{lem}
\[
\text{if\:\ensuremath{\vdash_{\boldsymbol{K_{TM}}^{\boldsymbol{\mathcal{A}},\overline{\sigma}}}}}\phi\rightarrow\psi\:\text{then}\:\vdash_{\boldsymbol{K_{TM}}^{\boldsymbol{\mathcal{A}},\overline{\sigma}}}K_{t}\phi\rightarrow K_{t}\psi
\]

\begin{proof}
Assume $\text{\ensuremath{\vdash_{\boldsymbol{K_{TM}}^{\boldsymbol{\mathcal{A}},\overline{\sigma}}}}}\phi\rightarrow\psi$.
By KG we obtain $\text{\ensuremath{\vdash_{\boldsymbol{K_{TM}}^{\boldsymbol{\mathcal{A}},\overline{\sigma}}}}}K_{t}\big(\phi\rightarrow\psi\big)$
and thus K yields the desired result.
\end{proof}
The last Lemma to be proven here is a small result that is simply
practical to have at hand.
\begin{lem}
\emph{\label{Lemma Consistency}{[}Consistency Lemma{]} Let $\Gamma$
be any $\Lambda$ \textendash{} consistent set of formulas, and $\phi$
any formula. Then either is $\Gamma\cup\{\phi\}$ consistent or $\Gamma\cup\{\neg\phi\}$
is.}
\end{lem}
\begin{proof}
Assume that $\Gamma\cup\{\phi\}$ is inconsistent. As $\Gamma$ is
consistent this means that $\Gamma\vdash_{\Lambda}\neg\phi$ but then
the set \emph{$\Gamma\cup\{\neg\phi\}$ }is consistent.
\end{proof}
We are now ready to undertake the work of showing soundness, i.e that
any Theorem of our logic is \emph{valid} wrt. the class of all $TM_{\boldsymbol{\mathcal{A}}}^{\overline{\sigma}}$
\textendash{} frames, or equivalently containment of the logic specified
via syntax in the logic we have specified semantically.

\section{Soundness}

We start by the central definition:
\begin{defn}
\label{Def. Soundness}{[}Soundness{]} Let $S$ be a class of structures
(either models or frames), and let $\Lambda$ be a logic. We say that
$\Lambda$ is \emph{sound wrt. $S$ }if for every formula $\phi$
of the language we have
\end{defn}
\[
\text{if}\:\vdash_{\Lambda}\phi\:\text{then\:}\models_{S}\phi.
\]

\noindent And we go about showing soundness for $\boldsymbol{K_{TM}}^{\boldsymbol{\mathcal{A}},\overline{\sigma}}$wrt.
the class $F$ of all $TM_{\boldsymbol{\mathcal{A}}}^{\overline{\sigma}}$
\textendash{} frames in two steps: First we show that all axioms are
valid, and then we show that validity is preserved by each of the
rules of inference. 

\subsection{Validity of Axioms}

We have to show that all of the axioms $\boldsymbol{\forall}$, \textbf{Id},
\textbf{PS},\textbf{$\exists$Id}, \textbf{N}, \textbf{K}, \textbf{BF},
and \textbf{KNI} are valid. In the following let $\mathcal{F}=\langle\mathcal{W},DOM,\mathcal{R\rangle}$.

\medskip{}

\noindent $\boldsymbol{\forall}:$ We have to show that $\models_{F}\forall x\phi\rightarrow\phi(y/x)$
where $x,y\in VAR$ and $y$ has only free occurrences in $\phi$
so assume $\mathcal{M},w\models_{v}\forall x\phi$ for arbitrary $\mathcal{M}\in F$,
$w\in\mathcal{W}$, and $v$ a valuation. Let $v'$ be an $x$ \textendash{}
variant of $v$ such that $v(x)=v'(y)$ such that we have $\mathcal{M},w\models_{v'}\phi$
by the semantics for $\forall$. By Proposition \ref{Prop. Principle of replacement}
we obtain $\mathcal{M},w\models_{v}\phi(y/x)$, and as $\mathcal{M}$,
$w$ and $v$ was arbitrary this yields the validity of the axiom
$\forall$ with respect to $F$.

\medskip{}

\noindent \textbf{Id:} We have to show that where $t$ is any term
we have $\models_{F}(t=t)$, but this is immediately clear from the
semantics of $=$.

\medskip{}

\noindent \textbf{MSC: }By the semantics of $=$ and the definition
of the valuation \textbf{MSC }is trivially valid.

\medskip{}

\noindent \textbf{PS:} We have to show that where $x$ and $y$ are
variables and $\phi$ any wff we have $\models_{F}(x=y)\rightarrow\big(\phi(x)\rightarrow\phi(y)\big)$,
which will be accomplished by induction on the complexity of $\phi$.
Assume $\mathcal{M},w\models_{v}(x=y)$ and $\mathcal{M},w\models_{v}\phi(x)$
for arbitrary $\mathcal{M}\in F$, $w\in\mathcal{W}$ and valuation
$v$. By the semantics for $=$ we immediately get $v(x)=v(y)$ and
thus if $\phi$ is atomic $\models_{v}\phi(y)$ follows readily. This
establishes the induction base. Assume now that we have the result
for $\psi$ and wish to show \textbf{PS} for$\neg\psi$. If $\mathcal{M},w\models_{v}(x=y)$
we get by the induction hypothesis that $\models_{v}\big(\neg\psi(x)\leftrightarrow\neg\psi(y)\big)$
and as we by assumption have that $\models_{v}\neg\psi(x)$ it yields
a contradiction if $\models_{v}\psi(y)$. As $\mathcal{M}$, $w$
and $v$ was arbitrary we conclude that \textbf{PS} holds for $\neg\psi$.

If we know \textbf{PS }for formulas $\phi$ and $\psi$ we show that
it holds for $\phi\rightarrow\psi$, so assume $\mathcal{M},w\models_{v}(x=y)$
and $\mathcal{M},w\models_{v}\big(\phi(x)\rightarrow\psi(x)\big)$.
Now, \emph{either} $\mathcal{M},w\models_{v}\neg\phi(x)$ \emph{or}
$\mathcal{M},w\models_{v}\psi(x)$ by the semantics for $\rightarrow$,
but since \textbf{PS }holds for $\neg\phi$ and $\psi$ by the previous
and the induction hypothesis this means that either $\mathcal{M},w\models_{v}\neg\phi(y)$
or $\mathcal{M},w\models_{v}\psi(y)$ which in turn implies $\mathcal{M},w\models_{v}\phi(y)\rightarrow\psi(y)$.
As $\mathcal{M}$, $w$ and $v$ was arbitrary we conclude validity
of \textbf{PS }for $\phi\rightarrow\psi$.

If we assume \textbf{PS }for $\phi$ it vacously follows that \textbf{PS
}holds for $\forall x\phi$ so we turn to the modal case.

Assume \textbf{PS }for $\phi$ and let $t\in Term_{agt}$. We show
\textbf{PS }for $K_{t}\phi$. To this end, assume that $\mathcal{M},w\models_{v}(x=y)$
for variables $x,y\in VAR_{agt}$\linebreak{}
 and $\mathcal{M},w\models_{v}K_{x}\phi(x)$. But by the semantics
of $K_{t}$ this means that $\mathcal{M},w'\models_{v}\phi(x)$ for
any $w'$ such that $(w,w')\in\mathcal{R}(x^{w,v})$, but then the
induction hypothesis yields that $\mathcal{M},w'\models_{v}\phi(y)$
and as by assumption $x^{w,v}=y^{w,v}$ we get $\mathcal{M},w'\models_{v}\phi(y)$
for any $w'$ such that $(w,w')\in\mathcal{R}(y^{w,v})$. This is
exactly the definition of\linebreak{}
 $\mathcal{M},w\models_{v}K_{y}\phi(y)$ as wanted.

\medskip{}

\noindent \textbf{$\exists$Id: }We need to show that for any constant
$c$, any $\mathcal{M}\in F$, $w\in\mathcal{W}$ and valuation $v$
we have $\mathcal{M},w\models_{v}(c=c)\rightarrow\exists x(x=c)$
so assume $\mathcal{M},w\models_{v}(c=c)$. By the semantics of $\forall$
and the definition of $\exists$ we need to produce a valuation $v'$
such that $v'(x)=c^{w,v}$ but since all valuations are surjective
by definition this is indeed possible for all constants $c$ and all
$w\in\mathcal{W}$.

\medskip{}

\noindent \textbf{N: }This is simply true by construction. The axiom
says that there are exactly $|\boldsymbol{\mathcal{A}}|=n$ elements
which are quantified over using variables from $VAR_{agt}$, but this
is precisely the content of definition \ref{Def. Frame} 3. 

\medskip{}

\noindent \textbf{K: }We have to show that for any $\mathcal{M}\in F$,
$w\in\mathcal{W}$, valuation $v$, agent-denoting term $t$, and
wffs $\phi$ and $\psi$ we have $\mathcal{M},w\models_{v}K_{t}\big(\phi\rightarrow\psi\big)\rightarrow\big(K_{t}\phi\rightarrow K_{t}\psi\big)$,
so assume $\mathcal{M},w\models_{v}K_{t}\big(\phi\rightarrow\psi\big)$.
By the semantics of $K_{t}$ this means that $\mathcal{M},w'\models_{v}\phi\rightarrow\psi$
for every $w'\in\mathcal{W}$ with $(w,w')\in\mathcal{R}(t^{w,v})$
but this means that for all such $w'$ we have either $\mathcal{M},w'\models_{v}\neg\phi$
or $\mathcal{M},w'\models_{v}\psi$. All in all, we either have $\mathcal{M},w'\models_{v}\psi$
for all $w'\in\mathcal{W}$ with the property that $(w,w')\in\mathcal{R}(t^{w,v})$
in which case it holds that $\mathcal{M},w\models_{v}K_{t}\phi$,
or we have for some such $w'$ that $\grave{\mathcal{M},w'\models_{v}\neg\phi}$
in which case $\mathcal{M},w\models_{v}\neg K_{t}\phi$. By the semantics
of $\rightarrow$ we conclude $\mathcal{M},w\models_{v}K_{t}\phi\rightarrow K_{t}\psi$
which was what we wanted. 

\medskip{}

\noindent \textbf{BF: }We need to show that for any $\mathcal{M}\in F$,
$w\in\mathcal{W}$, valuation $v$, agent-denoting term $t$, variable
$x\neq t$, and wff $\phi$ we have $\mathcal{M},w\models_{v}\forall xK_{t}\phi\rightarrow K_{t}\forall x\phi$.
Assume for contradiction that $\mathcal{M},w\models_{v}\forall xK_{t}\phi$
yet $\mathcal{M},w\models_{v}\neg K_{t}\forall x\phi$. By the latter
there is some world $w'$ for which $(w,w')\in\mathcal{R}(t^{w,v})$
such that $\mathcal{M},w'\models_{v}\neg\forall x\phi$ which in turn
means that there is $v'$ an $x$ \textendash{} variant of $v$ such
that $\mathcal{M},w'\models_{v'}\neg\phi$. However, by the former
we get that $\mathcal{M},w\models_{v'}K_{t}\phi$ and thus $\mathcal{M},w'\models_{v'}\phi$,
yielding a contradiction.

\medskip{}

\noindent \textbf{KNI: }We have to prove that for any $\mathcal{M}\in F$,
$w\in\mathcal{W}$, valuation $v$, agent-denoting term $t$, variables
$x$ and $y$, we have $\mathcal{M},w\models_{v}(x\neq y)\rightarrow K_{t}(x\neq y)$.
Assume $\mathcal{M},w\models_{v}(x\neq y)$. As variables are rigid
this will hold in any world, thus especially those worlds $w'\in\mathcal{W}$
such that $(w,w')\in\mathcal{R}(t^{w,v})$.

\medskip{}

\noindent Having dealt with all the axioms we turn to show that the
rules of inference \textbf{MP}, \textbf{KG}, and \textbf{Gen }preserves
validity.

\subsection{Rules of Inference Preserves Validity}

\textbf{MP: }We have to show that if $\phi\rightarrow\psi$ and $\phi$
are valid on $F$ then $\psi$ is valid on $F$, so assume that $\models_{F}\phi\rightarrow\psi$
and $\models_{F}\phi$. By the latter we see that $\phi$ is satisfied
in every model based on $F$, in every world, under every valuation.
But then the semantics for $\rightarrow$ yields that so must $\psi$
be which is what we wanted.

\medskip{}

\noindent \textbf{KG: }We have to show that if $\phi$ is valid on
$F$ then so is $K_{t}\phi$ for any agent-referring term $t$, but
this is immediately clear from the semantics for $K_{t}$.

\medskip{}

\noindent \textbf{Gen: }We have to show that if $\phi\rightarrow\psi$
is valid on $F$ for wffs $\phi$ and $\psi$ with the variable $x$
\emph{not }having any free occurrences in $\phi$ then $\phi\rightarrow\forall x\psi$
is also valid on $F$. Since $\phi\rightarrow\psi$ is valid on $F$
we know from the semantics for $\rightarrow$ that in every model
$\mathcal{M}\in F$, in every world $w\in\mathcal{W}$ for every valuation
$v$ we have either $\mathcal{M},w\models_{v}\neg\phi$ or $\mathcal{M},w\models_{v}\psi$
so the only way for \textbf{Gen }fail is if in some such model, world
and valuation we have $\mathcal{M},w\models_{v}\phi$, and $\mathcal{M},w\models_{v}\psi$
yet somehow $\mathcal{M},w\models_{v}\neg\forall x\psi$, but this
would mean that for some valuation $v'$, $x$ \textendash{} variant
of $v$, that $\mathcal{M},w\models_{v'}\neg\psi$ . As $x$ does
not occur free in $\phi$ $\mathcal{M},w\models_{v}\phi$ implies
$\mathcal{M},w\models_{v'}\phi$, contradicting the validity of $\phi\rightarrow\psi$. 

\noindent Having established validity of the axioms and the validity
preservation of the rules of inference soundness of $\boldsymbol{K_{TM}}^{\boldsymbol{\mathcal{A}},\overline{\sigma}}$
wrt. $F$ is a sitting duck:
\begin{thm}
\emph{\label{Theorem soundness}{[}Soundness{]} As all} $\boldsymbol{K_{TM}}^{\boldsymbol{\mathcal{A}},\overline{\sigma}}$\textendash{}
\emph{axioms are valid on the class $F$ of all} $TM_{\boldsymbol{\mathcal{A}}}^{\overline{\sigma}}$
\textendash{} \emph{frames, and all rules of inference preserves validity
we conclude that} $\boldsymbol{K_{TM}}^{\boldsymbol{\mathcal{A}},\overline{\sigma}}$
\emph{is sound wrt. $F$, i.e }if $\vdash_{\boldsymbol{K_{TM}}^{\boldsymbol{\mathcal{A}},\overline{\sigma}}}\phi$
then \emph{$\models_{F}\phi$.}
\end{thm}
\begin{proof}
See above.
\end{proof}
\noindent We now know, that the logic defined by the syntax presented
is included in the logic that follows from the semantic definitions,
that is $L_{\boldsymbol{K_{TM}}^{\boldsymbol{\mathcal{A}},\overline{\sigma}}}\subseteq L_{F}$
in the language of definitions \ref{Def. Semantically Specified Logic}
and \ref{Def. Proof Theoretically Specified Logic}. What remains
is the opposite inclusion which will be the topic of the next section.

\section{Completeness}

We start with the most central definition.
\begin{defn}
{[}Completeness{]} When $F$ is a class of structures (models or frames)
and $\Lambda$ is a logic, we say that $\Lambda$ is \emph{strongly
complete }wrt. $F$ if for any set $\Gamma$ of wffs, and $\phi$
a single wff, \emph{if }$\Lambda\models_{F}\phi$ \emph{then }$\Gamma\vdash_{\Lambda}\phi$.
That is, if $\phi$ is a semantic consequence of $\Gamma$ on $F$
then $\Gamma$ proves $\phi$ in $\Lambda$.
\end{defn}
Based on the proof of soundness one might infer that we should prove
completeness by manually show for each valid formula that a proof
existed. However, this will not be the strategy. Consider the following
Proposition.
\begin{prop}
\emph{\label{Prop. IFF}{[}IFF{]} The logic $\Lambda$ is strongly
complete wrt. the class of structures $F$ (models or frames) iff
any $\Lambda$ \textendash{} consistent set of formulas $\Gamma$
is satisfiable on some structure from $F$.}
\end{prop}
\begin{proof}
We first prove that \emph{if $\Lambda$ }is complete wrt. $F$, \emph{then
}every $\Lambda$ \textendash{} consistent set of formulas $\Gamma$
is satisfiable on some structure from $F$, so assume completeness
and pick a $\Lambda$\textendash{} consistent set of formulas $\Gamma\cup\{\phi\}$.
If\emph{ }$\Gamma\cup\{\phi\}$ is not satisfiable on any structure
from $F$ we have $\Gamma\models_{F}\neg\phi$ but then by completeness
we get $\Gamma\vdash_{\Lambda}\neg\phi$ meaning that $\Gamma\cup\{\phi\}$
was $\Lambda$\textendash{} inconsistent after all, a contradiction.

\noindent Now, we prove that \emph{if }any $\Lambda$ \textendash{}
consistent set of formulas is satisfiable on $F$, \emph{then} $\Lambda$
is complete wrt. $F$. Assume for contradiction that any $\Lambda$
\textendash{} consistent set of formulas is satisfiable on $F$, yet
$\Lambda$ is not complete wrt. $F$. This means for some set of wffs
$\Gamma\cup\{\phi\}$ that $\Gamma\models_{F}\phi$ yet $\Gamma\not\vdash_{\Lambda}\phi$.
By the latter it follows that the set $\Gamma\cup\{\neg\phi\}$ is
consistent and so by assumption satisfiable on some structure from
$F$, but this contradicts that $\Gamma\models_{F}\phi$.
\end{proof}
What is the significance of Proposition \ref{Prop. IFF}? It means
essentially that proving completeness is a question of model-hunting,
for all we need to do, given a consistent set of formulas, is to produce
a model (and a valuation) satisfying that set of formulas. This insight
has given rise to a more or less standardized construction of what
is called \emph{canonical models }and this will be the focus of the
present inquiry in the following. First, we need to produce a set
of worlds.

\subsection{Worlds of the Canonical Model}

Seeing that the overall aim is to produce, given a consistent set
of formulas, a model and a valuation satisfying that set of formulas
it is natural to let the worlds of the canonical model be \emph{sets
of consistent formulas }subject to appropriate extra conditions. Truth
will then be defined as membership. Imagine in the following that
a $\Lambda$\textendash{} consistent set $\Omega$ of formulas is
given, and that our task is to produce a model $\mathcal{M}_{\Omega}^{\Lambda}$
from $F$ satisfying $\Omega$. First we bring a couple of important
definitions.
\begin{defn}
\label{Def. Maximal}{[}Maximal $\Lambda$ \textendash{} Consistent{]}
Let $\Lambda$ be a logic and $\Gamma$ a set of wffs. We say that
$\Gamma$ is maximally $\Lambda$ \textendash{} consistent if $\Gamma$
is $\Lambda$ \textendash{} consistent \emph{and }no proper extension
of $\Gamma$ is $\Lambda$ \textendash{} consistent.
\end{defn}
We shall abbreviate such that we write that $\Gamma$ is a $\Lambda-MSC$
whenever $\Gamma$ is a maximally $\Lambda$ \textendash{} consistent
set. In order to make the machinery work we need to make sure that
whenever a formula of the form $\forall x\phi$ is \emph{not }included
in some world of the canonical model there must be some \emph{``witness''
}of the falsity. This motivates the next definition:
\begin{defn}
\label{Def. forall-property}{[}$\forall$ \textendash{} property{]}
If $\Gamma$ is a set of formulas, we say that it has the $\forall$
\textendash{} property if for every wff $\phi$, for every variable
$x$, there is some variable $y$ such that $\big(\phi(y/x)\rightarrow\forall x\phi\big)\in\Gamma$.
Note, that if some set $\Gamma$ of wffs has the $\forall$ \textendash{}
property then so will every set of wffs of which $\Gamma$ is a subset.
\end{defn}
What follows from definition \ref{Def. forall-property} is that if
for some $\Lambda-MSC$ $\Gamma$ $\phi$ we have $\forall x\phi\not\in\Gamma$,
then there is some variable $y$ such that $\phi(y/x)\not\in\Gamma$.
This is the reason why the $\forall$ \textendash{} property is sometimes
referred to as \emph{``bearing witness''}.

\noindent For technical reasons we need to enlarge our language in
order to make sure our construction will function, and so we add to
$\mathcal{L}_{TM}^{^{\boldsymbol{\mathcal{A}},\overline{\sigma}}}$
infinitely many \emph{new} variables equally divided between $VAR_{agt}$
and $VAR_{obj}$. The resulting language is called $\mathcal{L}^{+}$,
and note that trivially any wff of $\mathcal{L}_{TM}^{^{\boldsymbol{\mathcal{A}},\overline{\sigma}}}$
is also a wff of $\mathcal{L}^{+}$. Now we need to be better acquainted
with $\Lambda-MSC$'s, and for that we need a Lemma:
\begin{lem}
\emph{\label{Lemma completeness}{[}Completeness Lemma{]} Let $\Gamma$
be a consistent set of formulas such that for all formulas $\phi$
either $\phi\in\Gamma$ or $\neg\phi\in\Gamma$. Then $\Gamma$ is
deductively closed, that is whenever $\Gamma\vdash_{\Lambda}\psi$
we have $\psi\in\Gamma$.}
\end{lem}
\begin{proof}
Assume that $\Gamma$ has the mentioned property, and that for some
formula $\phi$ we have $\Gamma\vdash_{\Lambda}\phi$. By assumption
either $\phi\in\Gamma$ or $\neg\phi\in\Gamma$ but the latter contradicts
the consistency of $\Gamma$.
\end{proof}
\begin{lem}
\emph{\label{Lemma properties of MCS}{[}Properties of $\Lambda-MCS$'s{]}
If $\Lambda$ is a logic and $\Gamma$ is a $\Lambda-MCS$, then}
\end{lem}
\[
\begin{array}{rl}
i) & \Lambda\subseteq\Gamma\\
ii) & \text{for all formulas }\phi,\text{ either }\phi\in\Gamma\text{ or }\lnot\phi\in\Gamma\\
iii) & \Gamma\text{is deductively closed }\\
iv) & \text{for all formulas }\phi,\psi:(\phi\rightarrow\psi)\in\Gamma\text{ iff }\neg\phi\in\Gamma\text{ or }\psi\in\Gamma
\end{array}
\]

\begin{proof}
i) If $\phi\in\Lambda\setminus\Gamma$ then the sets $\Gamma\cup\{\phi\}$
and $\Gamma\cup\{\neg\phi\}$ are both proper extensions of $\Gamma$and
by Lemma \ref{Lemma Consistency} one of them is consistent contradicting
the maximality of $\Gamma$.

ii) Assume for contradiction that for some formula $\phi$ neither
$\phi\in\Gamma$ nor $\neg\phi\in\Gamma$, meaning that both $\Gamma\cup\{\phi\}$
and $\Gamma\cup\{\neg\phi\}$ are proper extensions of $\Gamma.$
By Lemma \ref{Lemma Consistency} one of them is consistent contradicting
the maximality of $\Gamma.$

iii) By assumption $\Gamma$ is consistent, and so it follows by ii)
and Lemma \ref{Lemma completeness} that $\Gamma$ is deductively
closed.

iv) Assume for contradiction that for some formulas $\phi$ and $\psi$\linebreak{}
 we have $\phi\rightarrow\psi\in\Gamma$ yet $\phi\in\Gamma$ \emph{and
}$\neg\psi\in\Gamma$ yielding $\{\big(\phi\rightarrow\psi\big),\phi,\neg\psi\}\subseteq\Gamma$.
By iii) this implies that $\psi\in\Gamma$ by an application of \textbf{MP,}
contradicting the consistency of $\Gamma$.
\end{proof}
The next needed result is the well-known Lindenbaum's Lemma. As the
proof is completely standard it will be omitted\footnote{For a proof see \cite{BlueModalLogic} or \cite{Rendsvig-StuS2010}.}.
\begin{lem}
\emph{\label{Lemma Lindenbaum}{[}Lindenbaum's Lemma{]} Let $\Gamma$
be a $\Lambda$ \textendash{} consistent set of formulas. Then there
exists some $\Lambda-MCS$ $\Gamma^{+}$such that $\Gamma\subseteq\Gamma^{+}.$}
\end{lem}
Given a consistent set of formulas $\Gamma$ we, in accordance with
Proposition \ref{Prop. IFF}, wish to produce a model and a valuation
satisfying $\Gamma$. We know by Lemma \ref{Lemma Lindenbaum} that
we can extend $\Gamma$ to a maximally consistent set, but $\Gamma$
is formulated over $\mathcal{L}_{TM}^{^{\boldsymbol{\mathcal{A}},\overline{\sigma}}}$
and the model we are about to build will be formulated over $\mathcal{L}^{+}$
so we need to ascertain ourselves that we can find an appropriate
set of formulas over $\mathcal{L}^{+}$extending $\Gamma$. This is
the content of the next Lemma:
\begin{lem}
\emph{\label{Lemma Saturation}{[}Saturation{]} Let $\Gamma$ be a
$\Lambda$ \textendash{} consistent set of formulas over} $\mathcal{L}_{TM}^{^{\boldsymbol{\mathcal{A}},\overline{\sigma}}}$.
\emph{Then there exist a $\Lambda$ \textendash{} consistent set $\Gamma_{+}$
over $\mathcal{L}^{+}$ with the $\forall$ \textendash{} property
such that $\Gamma\subseteq\Gamma_{+}$.}
\end{lem}
\begin{proof}
\emph{Let $\Gamma$ be a $\Lambda$ \textendash{} consistent set of
formulas over} $\mathcal{L}_{TM}^{^{\boldsymbol{\mathcal{A}},\overline{\sigma}}}$,
and observe that $\mathcal{L}^{+}$is countable and as any formula
is a finite string of symbols from $\mathcal{L}^{+}$the set of wffs
is also countable, meaning that we can enumerate all formulas of the
type $\forall x\phi$ over $\mathcal{L}^{+}$. Define now a sequence
of sets $\{\Delta_{i}\}_{i\in\mathbb{{N}}}$by

\begin{align*}
\Delta_{0} & :=\Gamma\\
\Delta_{n+1} & :=\Delta_{n}\cup\left\{ \varphi\left(y/x\right)\rightarrow\forall x\varphi\right\} 
\end{align*}

where we take $\phi$ to be the $n+1$'st formula wrt to the enumeration,
and $y$ to be the first variable(again, relative to the enumeration)
\emph{not to occur in }$\Delta_{n}$ nor $\phi$(and note that this
\emph{a fortiori }means that $y$ does not occur free anywhere in
neither $\Delta_{n}$ nor $\phi$). The reason why we introduced the
enlarged language $\mathcal{L}^{+}$was exactly to ensure that this
construction is possible; as $\Delta_{0}=\Gamma\subseteq\mathcal{L}_{TM}^{^{\boldsymbol{\mathcal{A}},\overline{\sigma}}}$
and only finitely many new variables are introduced in each step we
can always pick the variable $y\in\mathcal{L}^{+}$. 

\noindent From here the strategy is to show that each $\Delta_{i}$
is consistent and then choose our appropriate set of $\mathcal{\mathcal{L}}^{+}$
\textendash{} formulas as $\Gamma_{+}=\bigcup_{i\in\mathbb{N}}\Delta_{i}$.
We proceed by induction. 

\noindent To establish the induction base simply note that $\Gamma$
is $\Lambda$ \textendash{} consistent by assumption. Now, assume
that $\Delta_{n}$ is $\Lambda$ \textendash{} consistent while $\Delta_{n+1}$
is not. This means that for some $\phi_{1},...,\phi_{k}\in\Delta_{n}$
we have 
\begin{align*}
\vdash_{\Lambda}(\phi_{1}...\phi_{k})\rightarrow\phi(y/x)\: & \text{and}\\
\vdash_{\Lambda}(\phi_{1}...\phi_{k})\rightarrow & \neg\forall x\phi
\end{align*}

As $y$ does not occur free in any of the $\phi_{i}$'s we get from
the first above and \textbf{Gen }that $\vdash_{\Lambda}(\phi_{1}...\phi_{k})\rightarrow\forall y\phi(y/x)$
and since $y$ does not occur free in $\phi$ this is equivalent to 

\footnote{See \cite{Hughes_&_Cresswell_New.Intro} pp. 242, 258.}$\vdash_{\Lambda}(\phi_{1}...\phi_{k})\rightarrow\forall x\phi(x)$,
but this means that 
\begin{align*}
\vdash_{\Lambda}(\phi_{1}...\phi_{k})\rightarrow\forall x\phi\: & \text{and}\\
\vdash_{\Lambda}(\phi_{1}...\phi_{k})\rightarrow & \neg\forall x\phi
\end{align*}

contradicting the consistency of $\Delta_{n}$. We conclude that $\Delta_{i}$
is $\Lambda$ \textendash{} consistent for every $i\in\mathbb{N}$.
Put $\Gamma_{+}=\bigcup_{i\in\mathcal{\mathbb{N}}}\Delta_{i}$ and
observe that $\Gamma_{+}$ trivially has the $\forall$ \textendash{}
property by construction so it remains to show that consistency is
\emph{``preserved in the limit''}.

\noindent For $\Gamma_{+}$ to be inconsistent there would have to
exist some inconsistent, finite subset $\Gamma_{0}\subseteq\Gamma_{+}$
by definition \ref{Def. Consistency}, but any such $\Gamma_{0}$
would be contained in some $\Delta_{i}$, a contradiction. We conclude
that $\Gamma_{+}$ is $\Lambda$ \textendash{} consistent and has
the $\forall$ \textendash{} property.
\end{proof}
Note that the Saturation Lemma does not explicitly ensure preservation
of the $\forall$ \textendash{} property; this is not necessary as
a set of formulas has the $\forall$ \textendash{} property if any
subset has it. We are almost ready to actually define \emph{the canonical
model} but one more technicality is still pressing. Since variables
are rigid and we have included equality in our language we need to
make sure that the same equality statements are true in each world
in the canonical model. To ensure that this is the case we shall restrict
attention to a \emph{cohesive subset }of the set of worlds which motivates
the following definition:
\begin{defn}
\label{Def. R-path connected}{[}$\mathcal{R}$ \textendash{} path
connected{]} Given sets $\mathcal{X}$ and $\mathcal{Y}$ and a map\linebreak{}
 $\mathcal{\mathcal{R}}:\mathcal{Y}\rightarrow\mathcal{P}(\mathcal{X}\times\mathcal{X})$
associating binary relations on $\mathcal{X}$ to each member of $\mathcal{Y}$,
we say that the subset $\mathcal{B}\subseteq\mathcal{X}$ is $\mathcal{R}$
\textendash{} path connected if for every $b,b'\in\mathcal{B}$ there
is a sequence $(y_{1},...,y_{k})\subseteq\mathcal{Y}$, and a sequence
$(b_{1},...,b_{k-1})\subseteq\mathcal{B}$ such that $(b,b_{1})\in\mathcal{R}(y_{1})$,
$(b_{1},b_{2})\in\mathcal{R}(y_{2})$,..., $(b_{k-1},b')\in\mathcal{R}(y_{k})$.
\end{defn}
Note that if $\mathcal{M}$ is a model, then so is any $\mathcal{R}$
\textendash{} path connected subset $\mathcal{M}'\subseteq\mathcal{M}$
since by definition truth in $\mathcal{M}'$ cannot depend on worlds
in $\mathcal{M}\setminus\mathcal{M}'$. We are now ready to define
the \emph{canonical model }for a $\Lambda$ \textendash{} consistent
set $\Omega$.

\subsection{Canonical Models}
\begin{defn}
\label{Def. Canonical model}{[}Canonical Model{]} Given a $\Lambda$
\textendash{} consistent set $\Omega$ from an $n$ \textendash{}
agent, two-sorted term-modal logic $\Lambda$ in the language $\mathcal{L}_{TM}^{^{\boldsymbol{\mathcal{A}},\overline{\sigma}}}$
with extension $\mathcal{L}^{+}$, we define the \emph{canonical model
}$\mathcal{M}_{\Omega}^{^{\Lambda}}$ to be the four-tuple $\left\langle \mathcal{W}_{\Omega}^{\Lambda},Dom^{\Lambda},\mathcal{R}^{\Lambda},\mathcal{I}^{\Lambda}\right\rangle $
where
\end{defn}
\begin{enumerate}
\item $\mathcal{W}_{\Omega}^{\Lambda}$ is an $\mathcal{R}^{\Lambda}$ \textendash{}
path connected subset of all $\Lambda$\textendash{} MCS's with the
$\forall$-property in $\mathcal{L}^{\textbf{+}}$ containing some
$\Lambda$\textendash{} MCS extending $\Omega$.
\item $Dom^{\Lambda}=\left\{ \left[x\right]\:|\:x\in VAR_{agt}^{+}\right\} \dot{\bigcup}\left\{ \left[y\right]\:|\:y\in VAR_{obj}^{+}\right\} =\boldsymbol{\mathcal{A}}^{\Lambda}\dot{\bigcup}\left\{ \left[y\right]\:|\:y\in VAR_{obj}^{+}\right\} $
with $\left[z\right]=\left\{ v\in VAR^{+}\:|\:z\sim v\right\} $ where
$z\sim v$ iff $\left(z=v\right)\in w$ for any one $w\in\mathcal{W}_{\Omega}^{\Lambda}.$
Note that we indeed have a partition of the domain by \textbf{MSD.}
\item $\mathcal{R}^{\Lambda}:\boldsymbol{\mathcal{A}}^{\Lambda}\rightarrow\mathcal{P}(\mathcal{W}_{\Omega}^{\Lambda}\times\mathcal{W}_{\Omega}^{\Lambda})$
is a map associating binary accessibility relations to agents such
that $\left(w,w^{\prime}\right)\in\mathcal{R}^{\Lambda}(\alpha_{i}^{\Lambda})$
iff for every formula $K_{x}\varphi\in\mathcal{L}^{\text{\textbf{+}}}$
with $x\in\alpha_{i}^{\Lambda}\in\boldsymbol{\mathcal{A}}^{\Lambda}$
if $K_{x}\varphi\in w,$ then $\varphi\in w^{\prime}.$
\item For each $\beta\in\overline{\sigma}_{\Lambda}^{k}$, $w\in\mathcal{W}_{\Omega}^{\Lambda}$
and $P$ some relation symbol with arity $\beta$ let $\mathcal{I}^{\Lambda}(P,w)=\{\big([x_{1}]...[x_{k}]\big)\in\prod_{i\in\{1...k\}}DOM_{\beta_{i}}^{\Lambda}\:|\:P(x_{1},...,x_{k})\in w\}$.
\linebreak{}
For each $\alpha\in\overline{\sigma}_{\Lambda}^{k+1}$, $w\in\mathcal{W}_{\Omega}^{\Lambda}$
and $f$ some function symbol with arity $\alpha$ let\linebreak{}
 $\mathcal{I}^{\Lambda}(f,w)=\{\big([x_{1}],...,[x_{k}],[x_{k+1}]\big)\in\prod_{i\in\{1...k+1\}}DOM_{\alpha_{i}}^{\Lambda}\:|\:\big(f(x_{1},...,x_{k})=x_{k+1}\big)\in w\}$.
For any world $w\in\mathcal{W}_{\Omega}^{\Lambda}$ take $\mathcal{I}^{\Lambda}(=,w)=\{(d,d)\:|\:d\in Dom^{\Lambda}\}.$
\linebreak{}
For any constant $c\in CON$, for any world $w\in\mathcal{W}_{\Omega}^{\Lambda}$,\linebreak{}
 $\mathcal{I}^{\Lambda}(c,w)=[x]\in Dom^{\Lambda}$ such that $(x=c)\in w$.
\end{enumerate}
Before defining the \emph{canonical valuation }we need to make sure
that the canonical model is well-defined. This motivates the following
Proposition:
\begin{prop}
\emph{\label{Prop. Canonical Interpretation}{[}Canonical Interpretation
of Constants{]} The canonical interpretation does indeed ascribe extensions
to all constants in $\mathcal{L}^{+}$, i.e. for all worlds }$w\in\mathcal{W}_{\Omega}^{\Lambda}$
\emph{for every constant $c\in CON$ there is $y\in VAR^{+}$ such
that $(y=c)\in w.$}
\end{prop}
\begin{proof}
By the $\forall$ \textendash{} property we have for any wff $\phi$
and every variable $x$ that there exists a variable $y$ such that
$\big(\phi(y/x)\rightarrow\forall x\phi\big)\in w$, and we recall
that $w$ is deductively closed by Lemma \ref{Lemma properties of MCS}.
By contraposition this yields $\big(\neg\forall x\phi\rightarrow\neg\phi(y/x)\big)\in w$
and thus by substituting $\psi:=\neg\phi$ that $\big(\neg\forall x\neg\psi\rightarrow\neg\neg\psi(y/x)\big)\in w$.
Canceling double-negation and abbreviating using $\exists$ gets us
$\big(\exists x\psi\rightarrow\psi(y/x)\big)\in w$. By \textbf{Id},
$\exists$\textbf{Id }and \textbf{MP }we get $\big(\exists x\psi\big)\in w$
and thus by another application of \textbf{MP }that $\psi(y/x)\in w$.
Putting $\psi:=(x=c)$ now yields $(y=c)\in w$ as desired\footnote{Strictly speaking the proof uses that there is a one-to-one correspondence
between formulas $\phi$ and $\neg\phi$ but this is trivial. Note
also that contraposition and tertium non datur are substitution instanses
of valid formulas from propositional modal logic, and thus at our
disposal. }.
\end{proof}
The only thing left to specify is the canonical valuation:
\begin{defn}
\label{Def. Canonical valuation}{[}Canonical Valuation{]} The Canonical
Valuation $v^{\Lambda}$ is defined by $v^{\Lambda}(x)=[x]$ for all
$x\in VAR^{+}.$
\end{defn}
By construction the domain of the Canonical Model is depending on
what equalities between variables are to be found in the $\Lambda-MCS$'s
that constitutes the worlds. As such one might be a little queasy
that maybe the occurrence of different such equalities in different
worlds might pose incommensurable claims on the cardinality of the
domain. That this \emph{is not} a problem needs to be proven:
\begin{prop}
\emph{\label{Prop. Well-defd domain}{[}Well \textendash{} Defined
Domain{]} Let $w$ and $w'$ be any two worlds from }$\mathcal{W}_{\Omega}^{\Lambda}$
\emph{and $x$ and $y$ any two variables from $\mathcal{L}^{+}$
we have $(x=y)\in w$ iff $(x=y)\in w'$.}
\end{prop}
\begin{proof}
This is where we need that $\mathcal{W}_{\Omega}^{\Lambda}$ is $\mathcal{R}^{\Lambda}$\textendash{}
path connected. Assume for $x,y\in VAR^{+}$ and world $w\in\mathcal{W}_{\Omega}^{\Lambda}$
that $(x=y)\in w$. As $\mathcal{W}_{\Omega}^{\Lambda}$ is $\mathcal{R}^{\Lambda}$\textendash{}
path connected we can find $(\alpha_{i}^{\Lambda},...,\alpha_{k}^{\Lambda})\subseteq\boldsymbol{\mathcal{A}}^{\Lambda}$
and $(w_{1},...,w_{k-1})\subseteq\mathcal{W}_{\Omega}^{\Lambda}$
such that $(w,w_{1})\in\mathcal{R}(\alpha_{1}^{\Lambda})$, $(w_{1},w_{2})\in\mathcal{R}(\alpha_{2}^{\Lambda})$,...,
and $(w_{k-1},w')\in\mathcal{R}(\alpha_{k}^{\Lambda})$. Furthermore,
let the terms $t_{1},...,t_{k}\in Term_{agt}^{\Lambda}$ be such that
$t_{i}$ designates $\alpha_{i}^{\Lambda}$ for $i=1,...,k$. As worlds
are deductively closed cf. Lemma \ref{Lemma properties of MCS} we
get by an application of \textbf{KG }that $K_{t_{1}}(x=y)\in w$ and
thus by the definition of $\mathcal{R}^{\Lambda}$ we have $(x=y)\in w'$.
Repeating this procedure $k$ times yields the desired result.
\end{proof}
In the definition of the canonical model we made sure that whenever
a formula of the form $K_{t}\phi$ where contained in some world $w$,\linebreak{}
 we had $(w,w')\in\mathcal{R}^{\Lambda}(t^{w,v^{\Lambda}})$ only
when $\phi\in w'$. However, we also need to make sure that the interplay
between $\mathcal{R}^{\Lambda}$ and formulas of the type $P_{t}\phi$
is appropriate, i.e. if $P_{t}\phi\in w$ then there must be some
$w'$ such that $(w,w')\in\mathcal{R}^{\Lambda}(t^{w,v^{\Lambda}})$
with $\phi\in w'$. This is the content of the next Proposition:
\begin{prop}
\emph{\label{Prop. Existence}{[}Existence{]} If $w\in\mathcal{W}_{\Omega}^{\Lambda}$
such that $P_{x}\phi\in w$, then there is $w'\in\mathcal{W}_{\Omega}^{\Lambda}$
such that $(w,w')\in\mathcal{R}(\alpha_{i}^{\Lambda})$ and $\phi\in w'$.}
\end{prop}
\begin{proof}
The proof is constructive so assume $P_{x}\phi\in w$ for some world
\emph{$w\in\mathcal{W}_{\Omega}^{\Lambda}$ }and some variable $x$
with\emph{ $v^{\Lambda}(x)=\alpha_{i}^{\Lambda}$. }The aim is then
to produce a world $w'$ such that \emph{$(w,w')\in\mathcal{R}(\alpha_{i}^{\Lambda})$
}and\emph{ $\phi\in w'$. }Consider the set $\Gamma=\{\psi\}\cup\{\phi\:|\:K_{x}\phi\in w\}$.
We first show that $\Gamma$ is indeed $\Lambda$ \textendash{} consistent,
so assume otherwise. As $\{\phi\:|\:K_{x}\phi\in w\}$ is consistent
this means that there are $\phi_{1},...,\phi_{n}\in\{\phi\:|\:K_{x}\phi\in w\}$
such that 
\[
\vdash_{\Lambda}\big(\phi_{1}\wedge...\wedge\phi_{n}\big)\rightarrow\neg\psi
\]

\noindent by \textbf{KG }and two applications of \textbf{K }this yields
\[
\vdash_{\Lambda}\big(K_{x}\phi_{1}\wedge...\wedge K_{x}\phi_{n}\big)\rightarrow K_{x}\neg\psi.
\]

\noindent Since $w$ is an\emph{ $\Lambda-MCS$ }containing $K_{x}\phi_{i}$
for $i=1,...,n$ we can conclude that $\big(K_{x}\phi_{1}\wedge...\wedge K_{x}\phi_{n}\big)\in w$
and thus by \textbf{MP }also $K_{x}\neg\psi\in w$, but by the definition
of $P_{x}$ this means that $\neg P_{x}\psi\in w$ contradicting the
consistency of $w$. We conclude that the set $\Gamma$ is indeed
$\Lambda$ \textendash{} consistent. The idea from here is to construct
a sequence of formulas $\{\psi_{i}\}_{i\mathbb{\in N}}$ in a fitting
manner such that\linebreak{}
 $\{\bigcup_{i\in\mathbb{N}}\psi_{i}\}\cup\{\phi\:|\:K_{x}\phi\in w\}$
is a $\Lambda$ \textendash{} consistent set with the $\forall$ \textendash{}
property that can be extended by Lindenbaum's Lemma to an \emph{$\Lambda-MCS$}
with the required properties.

\noindent To define the sequence $\{\psi_{i}\}_{i\mathbb{\in N}}$
we first define a couple of enumerations; first we enumerate all formulas
of the type $\forall y\lambda$, and then letting $\forall y\lambda$
be the $n+1$'th such formula we define the sequence $\{\psi_{i}\}_{i\in\mathbb{N}}$
inductively by $\psi_{0}=\psi$ (from $\Gamma$) and $\psi_{n+1}=\psi_{n}\wedge\{\lambda(z/y)\rightarrow\forall y\lambda\}$
where $z\neq x$ is a variable such that\linebreak{}
 $\{\phi\:|\:K_{x}\phi\in w\}\cup\big(\psi_{n}\wedge\{\lambda(z/y)\rightarrow\forall y\lambda\}\big)$
is consistent. That it is always possible to choose such variable
$z$ is now shown:

\noindent Assume that $\{\psi_{n}\}\cup\{\phi\:|\:K_{x}\phi\in w\}$
is consistent and assume for contradiction that $\{\phi\:|\:K_{x}\phi\in w\}\cup\big(\psi_{n}\wedge\{\lambda(z/y)\rightarrow\forall y\lambda\}\big)$
is inconsistent for every variable $z$ of $\mathcal{L}^{+}$. This
means that for every variable $z$ we could find $\{\chi_{1},...,\chi_{n}\}\subseteq\{\phi\:|\:K_{x}\phi\in w\}$
such that 
\[
\vdash_{\Lambda}\big(\chi_{n}\wedge...\wedge\chi_{n}\big)\rightarrow\neg\big(\big(\psi_{n}\wedge\{\lambda(z/y)\rightarrow\forall y\lambda\}\big)\big)
\]

\noindent which by standard logic is equivalent to 
\[
\vdash_{\Lambda}\big(\chi_{n}\wedge...\wedge\chi_{n}\big)\rightarrow\big(\psi_{n}\rightarrow\{\lambda(z/y)\wedge\neg\forall y\lambda\}\big)
\]

\noindent and by Lemma \ref{Derived rule} we get 
\[
\vdash_{\Lambda}K_{x}\big(\chi_{n}\wedge...\wedge\chi_{n}\big)\rightarrow K_{x}\big(\psi_{n}\rightarrow\{\lambda(z/y)\wedge\neg\forall y\lambda\}\big)
\]

\noindent further, by Lemma \ref{Lemma K-Dist} we get 
\[
\vdash_{\Lambda}\big(K_{x}\chi_{n}\wedge...\wedge K_{x}\chi_{n}\big)\rightarrow K_{x}\big(\psi_{n}\rightarrow\{\lambda(z/y)\wedge\neg\forall y\lambda\}\big).
\]

\noindent Seeing that $\{\chi_{1},...,\chi_{n}\}\subseteq w$ by definition
we get by Lemma \ref{Lemma properties of MCS} that\linebreak{}
 $\big(K_{x}\chi_{n}\wedge...\wedge K_{x}\chi_{n}\big)\subseteq w$
and thus by an application of \textbf{MP }that\linebreak{}
 $K_{x}\big(\psi_{n}\rightarrow\{\lambda(z/y)\wedge\neg\forall y\lambda\}\big)\subseteq w$
for every variable $z$ of the language. Now, if $z\neq x$ is some
variable\emph{ }occurring in neither $\lambda$ nor $\psi_{n}$ we
consider the formula

\noindent 
\[
\forall zK_{x}\big(\psi_{n}\rightarrow\neg(\lambda(z/y)\rightarrow\forall y\lambda)\big)
\]

\noindent and see that by the $\forall$ \textendash{} property the
formula there is some variable $z'$ such that

\noindent 
\[
K_{x}\big(\psi_{n}\rightarrow\neg(\lambda(z'/y)\rightarrow\forall y\lambda)\big)\rightarrow\big(\forall zK_{x}\big(\psi_{n}\rightarrow\neg(\lambda(z/y)\rightarrow\forall y\lambda)\big)\big)\in w.
\]

\noindent But we just showed that $K_{x}\big(\psi_{n}\rightarrow\{\lambda(z/y)\wedge\neg\forall y\lambda\}\big)\subseteq w$
for every $z\in VAR^{+}$ and so by an application of \textbf{MP }we
obtain

\noindent 
\[
\forall zK_{x}\big(\psi_{n}\rightarrow\neg(\lambda(z/y)\rightarrow\forall y\lambda)\big)\in w.
\]

\noindent Furthermore, as $x\neq z$ we see that 
\[
\forall zK_{x}\big(\psi_{n}\rightarrow\neg(\lambda(z/y)\rightarrow\forall y\lambda)\big)\rightarrow K_{x}\forall z\big(\psi_{n}\rightarrow\neg(\lambda(z/y)\rightarrow\forall y\lambda)\big)
\]

\noindent is an instance of \textbf{BF }and so included in $w$. An
application of \textbf{MP }yields 

\noindent 
\[
K_{x}\forall z\big(\psi_{n}\rightarrow\neg(\lambda(z/y)\rightarrow\forall y\lambda)\big)\in w
\]

\noindent and utilizing that whenever $x$ does not occur free in
$\phi$, $\forall x(\phi\rightarrow\psi)\rightarrow(\phi\rightarrow\forall x\psi)$
is a Theorem of first-order logic\footnote{See \cite{Hughes_&_Cresswell_New.Intro} pp. 242.}
we obtain that

\noindent 
\begin{equation}
K_{x}\big(\psi_{n}\rightarrow\forall z\neg(\lambda(z/y)\rightarrow\forall y\lambda)\big)\in w\label{(*)}
\end{equation}

\noindent since by assumption $z$ does not occur (free) in $\psi_{n}$.
Recall that whenever $y$ is not free in $\forall x\phi$, $\exists x\big(\phi(y/x)\rightarrow\forall x\phi\big)$
is a Theorem of first-order logic, from which we get

\noindent 
\begin{equation}
\exists z\big(\lambda(z/y)\rightarrow\forall y\lambda\big).\label{(**)}
\end{equation}

\noindent Seeing that \ref{(*)} can be re-written as $K_{x}\big(\psi_{n}\rightarrow\neg\exists z\big(\lambda(z/y)\rightarrow\forall y\lambda\big)\big)\in w$
we conclude that $K_{x}\neg\psi_{n}\in w$ contradicting the consistency
of $\{\psi_{n}\}\cup\{\phi\:|\:K_{x}\phi\in w\}$. We conclude that
it is indeed possible to choose $z\neq x$ such that

\noindent 
\begin{equation}
\{\phi\:|\:K_{x}\phi\in w\}\cup\big(\psi_{n}\wedge\{\lambda(z/y)\rightarrow\forall y\lambda\}\big)=\{\phi\:|\:K_{x}\phi\in w\}\cup\{\psi_{n+1}\}\label{(***)}
\end{equation}

\noindent is $\Lambda$ \textendash{} consistent. Now, consider the
set $\{\phi\:|\:K_{x}\phi\in w\}\cup\{\bigcup_{n\in\mathbb{N}}\psi_{n}\}$.
First, for all $n$, $\{\phi\:|\:K_{x}\phi\in w\}\cup\{\psi_{n}\}$
is consistent by the above. Second, we see that $\vdash_{\Lambda}\psi_{n}\rightarrow\psi_{m}$
for all $n\geq m$ such that $\{\phi\:|\:K_{x}\phi\in w\}\cup\{\bigcup_{n\in\mathbb{N}}\psi_{n}\}$
is consistent too. Further, by construction the set has the $\forall$
\textendash{} property, and so can be extended to an \emph{$\Lambda-MCS$
}by Lindenbaum - call this set $w'$. This is the world we set out
to produce; by construction $\phi\in w'$ and so we have by definition
that $(w,w')\in\mathcal{R}^{\Lambda}(x^{w,v^{\Lambda}})$.
\end{proof}
Having secured the needed-in-a-second existence-result another worry
presents itself; for if $K_{x}\phi\in w$ for some world $w$, some
wff $\phi$, and some agent-referring variable $x$ yet for some variable
$y$ \emph{designating the same agent as $x$ }we had\emph{ $K_{y}\phi\notin$w
}we have a problem. To see this note that by Lemma \ref{Lemma properties of MCS}
that $\neg K_{y}\phi\in w$ or equivalently $P_{y}\neg\phi\in w$,
but by the Proposition just proved we have a world $w'$ such that
$(w,w')\in\mathcal{R}^{\Lambda}(y^{w,v^{\Lambda}})$ and $\neg\phi\in w'$
contradicting the assumption that $K_{x}\phi\in w$. That this does
not obtain is the content of the next Proposition:
\begin{prop}
\emph{\label{Prop. Uniformity}{[}Uniformity{]} Let $w\in\mathcal{W}_{\Omega}^{\Lambda}$
such that $K_{x}\phi\in w$ for wff $\phi$ and variable $x$ with
$v^{\Lambda}(x)=\alpha_{i}^{\Lambda}\in\boldsymbol{\mathcal{A}}^{\Lambda}$.
Then for all variables $y$ with $v^{\Lambda}(x)=v^{\Lambda}(y)$
we have $K_{y}\psi\in w$.}
\end{prop}
\begin{proof}
Let\emph{ $w\in\mathcal{W}_{\Omega}^{\Lambda}$ }such that\emph{ $K_{x}\phi\in w$
}for wff\emph{ $\phi$} and variable\emph{ $x$ }with \emph{$v^{\Lambda}(x)=\alpha_{i}^{\Lambda}\in\boldsymbol{\mathcal{A}}^{\Lambda}$.}\linebreak{}
\emph{ }Further, let $y$ be such that\emph{ $v^{\Lambda}(x)=v^{\Lambda}(y)$,
}but this means that $[x]=[y]$ so by the definition of the canonical
model and Proposition \ref{Prop. Well-defd domain} we have $(x=y)\in w'$
for every \emph{$w'\in\mathcal{W}_{\Omega}^{\Lambda}$, }hence a fortiori
$(x=y)\in w$. By \textbf{PS }we get 

\noindent 
\[
(x=y)\rightarrow\big(K_{x}\phi\rightarrow K_{y}\phi\big)\in w
\]

\noindent and thus by \textbf{MP }that $\big(K_{x}\phi\rightarrow K_{y}\phi\big)\in w$.
By assumption we have $K_{x}\phi\in w$ and so one more application
of \textbf{MP }yields\textbf{ }$K_{y}\phi\in w$ as desired.
\end{proof}
Obviously, the idea is that ``truth is membership'' in the canonical
model, and we are now in position to prove that this is exactly the
case.
\begin{lem}
\emph{\label{Lemma truth}{[}Truth{]} For every $w\in\mathcal{W}_{\Omega}^{\Lambda}$,
for every $\phi\in\mathcal{L}^{+}$we have $\mathcal{M}_{\Omega}^{^{\Lambda}},w\models_{v^{\Lambda}}\phi$
iff $\phi\in w$.}
\end{lem}
As the proof of the truth Lemma will be by induction of the complexity
of $\phi$ we need to make this notion precise first.
\begin{defn}
\label{Def. Complexity}{[}Complexity{]} For any wff $\phi$ of the
language, for any variable $x$ and any term-referring term $t$,
the complexity $c(\phi)$ of $\phi$ is given by
\begin{align*}
c(\phi) & =0,\text{\text{\,\,\,\,\,\,for all atomic formulas }}\phi\\
c(\phi\rightarrow\psi) & =\max\{c(\phi),c(\psi)\}+1\\
c(\neg\phi) & =c(\forall x\phi)=c(K_{t}\phi)=c(\phi)+1
\end{align*}
\end{defn}
\begin{proof}
Proof of the Truth Lemma by induction on the complexity of $\phi$.

\medskip{}

\noindent \textbf{Equality: }If $t_{1},t_{2}$ are terms and\emph{
$w\in\mathcal{W}_{\Omega}^{\Lambda}$ }then we have $\mathcal{M}_{\Omega}^{^{\Lambda}},w\models_{v^{\Lambda}}(t_{1}=t_{2})$
iff $\big(t_{1}^{w,v^{\Lambda}},t_{2}^{w,v^{\Lambda}}\big)\in\mathcal{I}^{\Lambda}(=,w)$
which in turn holds iff $\big(t_{1}^{w,v^{\Lambda}}=t_{2}^{w,v^{\Lambda}}\big)$.
Now, by the definition of extensions in the canonical model this is
the case iff for some $x\in VAR^{+}$ we have $t_{1}^{w,v^{\Lambda}},t_{2}^{w,v^{\Lambda}}\in[x]$
iff $(t_{1}=x)\in w$ and $(t_{2}=x)\in w$ and as $w$ is deductively
closed $(t_{1}=t_{2})\in w$ as desired.

\medskip{}

\noindent \textbf{Atomic Formulas: }Let $P$ be a predicate symbol
of arity $\beta\in\overline{\sigma}^{\Lambda}$ and\linebreak{}
 $t_{1},...,t_{k}\in Term^{+}$such that $t_{i}$ is of sort $\beta_{i}$.
Then, $\mathcal{M}_{\Omega}^{^{\Lambda}},w\models_{v^{\Lambda}}P(t_{1},...,t_{k})$
iff $\big(t_{1}^{w,v^{\Lambda}},...,t_{k}^{w,v^{\Lambda}}\big)=\big([x_{1}],...,[x_{k}]\big)\in\mathcal{I}^{\Lambda}(P,w)$
which in turn holds iff\linebreak{}
 $P(x_{1},...,x_{k})\in w$ and as by assumption $(t_{i}=x_{i})\in w$
for $i=1,...,k$ this is equivalent to $P(t_{1},...,t_{k})\in w$
by deductive closedness.

\medskip{}

\noindent \textbf{Negation: }Assume the truth Lemma for a wff $\phi$,
then\textbf{ }we have\emph{ $\mathcal{M}_{\Omega}^{^{\Lambda}},w\models_{v^{\Lambda}}\neg\phi$
}iff \emph{not $\mathcal{M}_{\Omega}^{^{\Lambda}},w\models_{v^{\Lambda}}\phi$
}iff $\phi\notin w$ by the induction hypothesis. By \ref{Lemma properties of MCS}
we get $\neg\phi\in w$ as desired.

\medskip{}

\noindent \textbf{Implication: }Assume the truth Lemma for wffs $\phi,\psi$.
Then we have\linebreak{}
\textbf{ }\emph{$\mathcal{M}_{\Omega}^{^{\Lambda}},w\models_{v^{\Lambda}}\phi\rightarrow\psi$
}iff either $\mathcal{M}_{\Omega}^{^{\Lambda}},w\models_{v^{\Lambda}}\neg\phi$
or \emph{$\mathcal{M}_{\Omega}^{^{\Lambda}},w\models_{v^{\Lambda}}\psi$
}which in turn holds iff either $\neg\phi\in w$ or $\psi\in w$ by
the induction hypothesis. Then Lemma \ref{Lemma properties of MCS}
yields $\phi\rightarrow\psi\in w$ as desired.

\medskip{}

\noindent \textbf{Universal: }We first show that if $\mathcal{M}_{\Omega}^{^{\Lambda}},w\models_{v^{\Lambda}}\forall x\phi$
then $\forall x\phi\in w$, so assume truth Lemma for wffs of complexity
lower than $\forall x\phi$, and that $\mathcal{M}_{\Omega}^{^{\Lambda}},w\models_{v^{\Lambda}}\forall x\phi$.
Thus, for all $x$ \textendash{} variants $v^{\Lambda'}$, $\mathcal{M}_{\Omega}^{^{\Lambda}},w\models_{v^{\Lambda'}}\phi$.
Now, by the $\forall$ \textendash{} property there exists variable
$y_{0}$ such that $\big(\phi(y_{0}/x)\rightarrow\forall x\phi\big)\in w$
, and as $\mathcal{M}_{\Omega}^{^{\Lambda}},w\models_{v^{\Lambda'}}\phi$
holds for \emph{all }$x$ \textendash{} variants we can in particular
choose $v^{\Lambda'}$ such that $v^{\Lambda}(x)=v^{\Lambda'}(y_{0})$
but then we get by Proposition \ref{Prop. Principle of replacement}
that $\mathcal{M}_{\Omega}^{^{\Lambda}},w\models_{v^{\Lambda'}}\phi(y_{0}/x)$,
and by the induction hypothesis $\phi(y_{0}/x)\in w$. An application
of \textbf{MP }yields $\forall x\phi\in w$ as desired.

\noindent For the converse we proceed by contraposition, so assume
$\forall x\phi\notin w$. By Lemma \ref{Lemma properties of MCS}
we get $\neg\forall x\phi\in w$, and by a contrapositive application
of the $\forall$ \textendash{} property we get $\neg\phi(y_{0}/x)\in w$
for some variable $y_{0}$ and thus by the induction hypothesis $\mathcal{M}_{\Omega}^{^{\Lambda}},w\models_{v^{\Lambda}}\neg\phi(y_{0}/x)$.
This means for the specific $x$ \textendash{} variant $v^{\Lambda'}$
such that $v^{\Lambda'}(x)=v^{\Lambda}(y_{0})$ we get by Proposition
\ref{Prop. Principle of replacement} that $\mathcal{M}_{\Omega}^{^{\Lambda}},w\models_{v^{\Lambda'}}\neg\phi(x)$
and thus $\mathcal{M}_{\Omega}^{^{\Lambda}},w\models_{v^{\Lambda}}\exists x\neg\phi(x)$
or equivalently $\mathcal{M}_{\Omega}^{^{\Lambda}},w\models_{v^{\Lambda}}\neg\forall x\phi(x)$
as desired.

\medskip{}

\noindent \textbf{Modal: }We first show that if $K_{x}\phi\in w$
then $\mathcal{M}_{\Omega}^{^{\Lambda}},w\models_{v^{\Lambda}}K_{x}\phi$,
so assume the truth Lemma for wffs of lower complexity than $K_{x}\phi$,
and let $K_{x}\phi\in w$ for some variable $x$ such that $v^{\Lambda}(x)=\alpha_{i}^{\Lambda}$.
If $w'\in\mathcal{W}_{\Omega}^{\Lambda}$ is any world such that $(w,w')\in\mathcal{R}^{\Lambda}(\alpha_{i}^{\Lambda})$
we have by definition that $\phi\in w'$ and thus by the induction
hypothesis that $\mathcal{M}_{\Omega}^{^{\Lambda}},w'\models_{v^{\Lambda}}\phi$.
As this holds for every such $w'$ we conclude that $\mathcal{M}_{\Omega}^{^{\Lambda}},w\models_{v^{\Lambda}}K_{x}\phi$
as desired.

\noindent We prove the converse by contraposition so assume $K_{x}\phi\notin w$
for some variable $x$ such that $v^{\Lambda}(x)=\alpha_{i}^{\Lambda}$.
By Lemma \ref{Lemma properties of MCS} we thus have $\neg K_{x}\phi\in w$,
or equivalently $P_{x}\neg\phi\in w$. Now, Proposition \ref{Prop. Existence}
yields the existence of some world $w'\in\mathcal{W}_{\Omega}^{\Lambda}$
such that $\neg\phi\in w'$ such that $(w,w')\in\mathcal{R}^{\Lambda}(\alpha_{i}^{\Lambda})$.
By the induction hypothesis $\mathcal{M}_{\Omega}^{^{\Lambda}},w'\models_{v^{\Lambda}}\neg\phi$
and thus $\mathcal{M}_{\Omega}^{^{\Lambda}},w\models_{v^{\Lambda}}P_{x}\neg\phi$or
equivalently $\mathcal{M}_{\Omega}^{^{\Lambda}},w\models_{v^{\Lambda}}\neg K_{x}\phi$
as desired.
\end{proof}
Before moving on, let us ascertain what we have accomplished so far.
Given a $\Lambda$ \textendash{} consistent set $\Omega$ \emph{formulated
in language $\mathcal{L}_{TM}^{^{\boldsymbol{\mathcal{A}},\overline{\sigma}}}$
}we have produced a model \emph{formulated in language $\mathcal{L}^{+}$
}satisfying $\Omega$, but in order to be able to argue for completeness
via Proposition \ref{Prop. IFF} we have rather to produce a model
\emph{formulated in} $\mathcal{L}_{TM}^{^{\boldsymbol{\mathcal{A}},\overline{\sigma}}}$.
Luckily, this is not really a problem since all we have to do is to
restrict the canonical valuation $v_{|\mathcal{L}_{TM}^{^{\boldsymbol{\mathcal{A}},\overline{\sigma}}}}^{\Lambda}$,
and throw away the excess variables from the language. To convince
oneself that this is legitimate simply consider why we introduced
new variables in the first place; they were never to be used \emph{as
part of the language }but rather as part of the domain. Having acquainted
ourselves with the canonical model we shall zoom out; instead of considering,
for each $\Lambda$ \textendash{} consistent set $\Omega$ a model,
we shall rather consider the class of such models. This motivates
the following definition:
\begin{defn}
\label{Def. Canonical class}{[}Canonical Class{]} Let $\Lambda$
be a normal two-sorted term-modal logic. Then we define \emph{the
class $M^{\Lambda}$ of canonical models }as the set of models $\mathcal{M}_{\Omega}^{\Lambda}$
for each $\Lambda$ \textendash{} consistent set $\Omega$.

\newpage{}

\noindent Which brings us to the main result:
\end{defn}
\begin{thm}
\emph{\label{Theorem Canonical class}{[}Canonical Class Theorem{]}}
\emph{Let }$\Lambda$ \emph{be a normal two-sorted term-modal logic.
Then $\Lambda$ is complete wrt. the class $M^{\Lambda}$ of canonical
models for $\Lambda$.}
\end{thm}
\begin{proof}
By Proposition \ref{Prop. IFF} proving completeness of $\Lambda$
wrt. $M^{\Lambda}$is simply a question of, given some $\Lambda$
\textendash{} consistent set $\Omega$, producing some element in
$M^{\Lambda}$on which $\Omega$ is satisfied. By Lemmas \ref{Lemma Lindenbaum}
and \ref{Lemma Saturation} we can extend $\Omega$ to an $\Lambda-MCS$
$w$ with the $\forall$ \textendash{} property, and we can find a
model $\mathcal{M}_{\Omega}^{\Lambda}\in M^{\Lambda}$ such that $\Omega\subseteq w\in M_{\Omega}^{\Lambda}$
which by Lemma \ref{Lemma truth} gives that $M_{\Omega}^{\Lambda},w\models_{v^{\Lambda}}\Omega$
. As this holds for every $\Lambda$ \textendash{} consistent set
$\Omega$ we conclude that \emph{$\Lambda$ }is complete wrt. the
class\emph{ $M^{\Lambda}$ }of canonical models.
\end{proof}
A corollary of this is that $\boldsymbol{K_{TM}}^{\boldsymbol{\mathcal{A}},\overline{\sigma}}$
is complete wrt. the class of all $TM_{\boldsymbol{\mathcal{A}}}^{\overline{\sigma}}$
\textendash{} frames.
\begin{cor}
\emph{\label{Cor. Complete ALL}The logic} $\boldsymbol{K_{TM}}^{\boldsymbol{\mathcal{A}},\overline{\sigma}}$
\emph{is complete wrt. the class of all }$TM_{\boldsymbol{\mathcal{A}}}^{\overline{\sigma}}$
\textendash{}\emph{ frames}.
\end{cor}
\begin{proof}
By Proposition \ref{Prop. IFF} we have to produce, for each $\boldsymbol{K_{TM}}^{\boldsymbol{\mathcal{A}},\overline{\sigma}}$\textendash{}
consistent set $\Omega$ a world $w$ in a model $\mathcal{M}$ based
on some\emph{ }$TM_{\boldsymbol{\mathcal{A}}}^{\overline{\sigma}}$
\textendash{}\emph{ }frame\emph{ }and a valuation $v$ such that $\mathcal{M},w\models_{v}\Omega$.
Now, choose the model to be $\mathcal{M}_{\Omega}^{^{\boldsymbol{K_{TM}}^{\boldsymbol{\mathcal{A}},\overline{\sigma}}}}$,
the world $w$ to be some $\Omega$ \textendash{} extending $\boldsymbol{K_{TM}}^{\boldsymbol{\mathcal{A}},\overline{\sigma}}-MCS$,
and the valuation to be $v^{\boldsymbol{K_{TM}}^{\boldsymbol{\mathcal{A}},\overline{\sigma}}}$.
Now we have $\mathcal{M}_{\Omega}^{^{\boldsymbol{K_{TM}}^{\boldsymbol{\mathcal{A}},\overline{\sigma}}}},w\models_{v^{\boldsymbol{K_{TM}}^{\boldsymbol{\mathcal{A}},\overline{\sigma}}}}\Omega$
by Lemma \ref{Lemma truth}, and we conclude that $\boldsymbol{K_{TM}}^{\boldsymbol{\mathcal{A}},\overline{\sigma}}$
is complete wrt. the class of all\emph{ }$TM_{\boldsymbol{\mathcal{A}}}^{\overline{\sigma}}$
\textendash{}\emph{ }frames by Theorem \ref{Theorem Canonical class}.
\end{proof}
Conclusively, $\boldsymbol{K_{TM}}^{\boldsymbol{\mathcal{A}},\overline{\sigma}}$
is sound and complete wrt. the class of all $TM_{\boldsymbol{\mathcal{A}}}^{\overline{\sigma}}$
\textendash{}\emph{ }frames.

\section{Applications of the Canonical Class Theorem}

In this section applications of the Canonical Class Theorem will be
explored, and soundness and completeness for the term-modal version
of $\boldsymbol{S4}$ will be proved. It is a fact from standard modal
logic that the axioms $\boldsymbol{T}$ (i.e. $\forall x\big(K_{x}\phi\rightarrow\phi\big)$),
$\boldsymbol{5}$ (i.e. $\forall x\big(P_{x}\phi\rightarrow K_{x}P_{x}\phi\big)$)
and $\boldsymbol{4}$ (i.e. $\forall x\big(K_{x}\phi\rightarrow K_{x}K_{x}\phi\big)$)
characterizes the class of reflexive, euclidian, and transitive frames
respectively (cf. \cite[pp. 128]{BlueModalLogic}) and the appropriate
term-modal version of these results will now be stated and proved.
\begin{lem}
\emph{\label{Lemma axiom T}{[}Axiom $\boldsymbol{T}${]} Let $x$
be any agent-referring variable, and $\phi$ any wff of $\mathcal{L}_{TM}^{^{\boldsymbol{\mathcal{A}},\overline{\sigma}}}$.
The axiom $\forall x\big(K_{x}\phi\rightarrow\phi\big)$ characterizes
the class of frames in which $\mathcal{R}(\alpha)$ is }reflexive
\emph{for all $\alpha\in\boldsymbol{\mathcal{A}}$.}
\end{lem}
\begin{proof}
Recall that a relation $R$ on a set $W$ is reflexive iff for all
$w\in W$ we have that $wRw$. We show first that if for some frame
$\mathcal{F}=\langle\mathcal{W},\mathcal{R},DOM\rangle$,\emph{$\mathcal{R}(\alpha)$
is }reflexive \emph{for all $\alpha\in\boldsymbol{\mathcal{A}}$ }then
the axiom\emph{ $\forall K_{x}\phi\rightarrow\phi$ }is valid on $\mathcal{F}$
for any agent-referring variable $x$. Fix some agent-referring variable
$x$, world $w\in\mathcal{W}$, and wff $\phi$. Assume further that
$\mathcal{R}(x^{w,v})$ is reflexive.\linebreak{}
It suffices to show that $\mathcal{M},w\models K_{x}\phi\rightarrow\phi$
by surjectivity of valutions, so assume $\mathcal{M},w\models K_{x}\phi$.
By reflexivity we have $(w,w)\in\mathcal{R}(x^{w,v})$ and so by the
semantics for $K_{x}$ that $\mathcal{M},w\models\phi$ as desired.

\noindent The converse is shown by contraposition: Assume $\mathcal{R}(\alpha)$
is non-reflexive for some agent $\alpha\in\boldsymbol{\mathcal{A}}$,
i.e. for some world $w\in\mathcal{W}$ we have $(w,w)\notin\mathcal{R}(\alpha)$.
Now we can pick the interpretation such that for the resulting model
$\mathcal{M}$, we have $\mathcal{M},w\not\models\phi$ while for
all $w'\in\mathcal{W}\setminus\{w\}$ we have $\mathcal{M},w'\models\phi$.
By the semantics of $K_{x}$ we get that $\mathcal{M},w\models K_{x}\phi$
yet $\mathcal{M},w\not\models\phi$.
\end{proof}
\begin{lem}
\emph{\label{Lemma axiom 5}{[}Axiom}\textbf{\emph{ }}\emph{$\boldsymbol{5}${]}
Let $x$ be any agent-referring variable, and $\phi$ any wff of $\mathcal{L}_{TM}^{^{\boldsymbol{\mathcal{A}},\overline{\sigma}}}$.
The axiom $\forall x\big(P_{x}\phi\rightarrow K_{x}P_{x}\phi\big)$
characterizes the class of frames in which $\mathcal{R}(\alpha)$
is }euclidian \emph{for all $\alpha\in\boldsymbol{\mathcal{A}}$.}
\end{lem}
\begin{proof}
Recall that a relation $R$ on a set $W$ is euclidian iff. for all\linebreak{}
 $v,u,w\in W$ if $uRv$ and $uRw$ then $vRw$. We show first that
if for some frame $\mathcal{F}=\langle\mathcal{W},\mathcal{R},DOM\rangle$,\emph{$\mathcal{R}(\alpha)$
is }euclidian \emph{for all $\alpha\in\boldsymbol{\mathcal{A}}$ }then
the axiom\emph{ $\forall x\big(P_{x}\phi\rightarrow K_{x}P_{x}\phi\big)$
}is valid on $\mathcal{F}$ for any agent-referring variable $x$,
so fix some agent-referring variable $x$, world $w\in\mathcal{W}$,
and any wff $\phi$. Assume further that $\mathcal{R}(x^{w,v})$ is
euclidian and see that it suffices to show that $\mathcal{M},w\models P_{x}\phi\rightarrow K_{x}P_{x}\phi$,
so assume $\mathcal{M},w\models P_{x}\phi$. By assumption there is
$w'\in\mathcal{W}$ such that $\mathcal{M},w'\models\phi$. If $w'$
is the only world accessible from $w$ we have by definition that
$M,w\models K_{x}P_{x}\phi$ so assume otherwise, i.e. there is some
$u\neq w'\in\textit{\ensuremath{\mathcal{W}}}$ s.t $(w,u)\in\mathcal{R}(x^{w,v}).$
As $\mathcal{R}(x^{w,v})$ is euclidian we get that $(u,w')\in\mathcal{R}(x^{w,v})$
and thus $M,u\models P_{x}\phi$ and then $M,w\models K_{x}P_{x}\phi$
as desired.

\noindent The converse is shown by contraposition, so let $\mathcal{R}(\alpha)$
is non-euclidian for some agent $\alpha\in\boldsymbol{\mathcal{A}}$,\\
 i.e. for worlds $w,v,u\in\mathcal{W}$ we have $(w,v)\in\mathcal{R}(\alpha)$
and $(w,u)\in\mathcal{R}(\alpha)$ while $(v,u)\notin\mathcal{R}(\alpha)$.
We can now choose our interpretation such that in the resulting model
$\mathcal{M}$, we have $\mathcal{M},v\models\phi$ yet $\mathcal{M},w\not\models\phi$
for every $w\in\mathcal{W}\setminus\{v\}$. Now, $\mathcal{M},w\models P_{x}\phi$
yet $\mathcal{M},w\not\models K_{x}P_{x}\phi$ as desired.
\end{proof}
\begin{lem}
\emph{\label{Lemma axiom 4}{[}Axiom $\boldsymbol{4}${]} Let $x$
be any agent-referring variable, and $\phi$ any wff of $\mathcal{L}_{TM}^{^{\boldsymbol{\mathcal{A}},\overline{\sigma}}}$.
The axiom $\forall x\big(K_{x}\phi\rightarrow K_{x}K_{x}\phi\big)$
characterizes the class of frames in which $\mathcal{R}(\alpha)$
is }transitive \emph{for all $\alpha\in\boldsymbol{\mathcal{A}}$.}
\end{lem}
\begin{proof}
Recall that a relation $R$ on a set $W$ is transitive iff. for all\linebreak{}
 $v,u,w\in W$ if $vRu$ and $uRw$ then $vRw$. We show first that
if for some frame $\mathcal{F}=\langle\mathcal{W},\mathcal{R},DOM\rangle$,\emph{$\mathcal{R}(\alpha)$
is }transitive \emph{for all $\alpha\in\boldsymbol{\mathcal{A}}$
}then the axiom\emph{ $\forall x\big(K_{x}\phi\rightarrow K_{x}K_{x}\phi\big)$
}is valid on $\mathcal{F}$ for any agent-referring variable $x$,
so fix some agent-referring variable $x$, world $w\in\mathcal{W}$,
and any wff $\phi$, and see that it suffices to show that $M,w\models K_{x}\phi\rightarrow K_{x}K_{x}\phi$,
so assume $\mathcal{M},w\models P_{x}\phi$. If $v,u\in\mathcal{W}$
is such that $(w,v)\in\mathcal{R}(\alpha)$ and $(v,u)\in\mathcal{R}(\alpha)$
we have by transitivity that $(w,u)\in\mathcal{R}(\alpha)$ and thus
by assumption that $\mathcal{M},u\models\phi$, such that $\mathcal{M},v\models K_{x}\phi$
meaning that $\mathcal{M},w\models K_{x}K_{x}\phi$ as desired.

\noindent The converse is shown by contraposition, so let $\mathcal{R}(\alpha)$
be non-transitive for some agent $\alpha\in\boldsymbol{\mathcal{A}}$,
i.e. for worlds $w,v,u\in\mathcal{W}$ we have $(w,v)\in\mathcal{R}(\alpha)$\linebreak{}
 and $(v,u)\in\mathcal{R}(\alpha)$ while $(w,u)\notin\mathcal{R}(\alpha)$.
We can now choose an interpretation such that in the resulting model,
we have $\mathcal{M},u\not\models\phi$ while $\mathcal{M},v\models\phi$
for all $v\in\mathcal{W}\setminus\{u\}$. By construction we have
$\mathcal{M},w\models K_{x}\phi$ yet $\mathcal{M},w\not\models K_{x}K_{x}\phi$
as desired.
\end{proof}
\begin{defn}
\label{Def. Logics}{[}$\boldsymbol{TM_{\overline{\sigma}}.K4,}$
$\boldsymbol{TM_{\overline{\sigma}}.K5},$ and $\boldsymbol{TM_{\overline{\sigma}}.KT}${]}
Denote by $\boldsymbol{TM_{\overline{\sigma}}.K4}$, $\boldsymbol{TM_{\overline{\sigma}}.K5},$
and $\boldsymbol{TM_{\overline{\sigma}}.KT}$ the logics resulting
from adding to $\boldsymbol{K_{TM}}^{\boldsymbol{\mathcal{A}},\overline{\sigma}}$
the axiom $\forall x\big(K_{x}\phi\rightarrow K_{x}K_{x}\phi\big)$,
\emph{$\forall x\big(P_{x}\phi\rightarrow K_{x}P_{x}\phi\big)$, }and\emph{
$\forall x\big(K_{x}\phi\rightarrow\phi\big)$ respectively} for any
agent-referring variable $x$ and wff $\phi$, and closing the resulting
collection of formulas under \textbf{MP}, \textbf{KG}, and \textbf{Gen.}
\end{defn}
\begin{cor}
\emph{\label{Cor. Soundness logics}{[}Soundness} $\boldsymbol{TM_{\overline{\sigma}}.K4,}$
$\boldsymbol{TM_{\overline{\sigma}}.K5},$ \emph{and }$\boldsymbol{TM_{\overline{\sigma}}.KT}$\emph{{]}
The logic $\boldsymbol{TM_{\overline{\sigma}}.K4}$ is sound wrt.
the class of transitive} $TM_{\boldsymbol{\mathcal{A}}}^{\overline{\sigma}}$
\textendash{} \emph{frames, the logic }$\boldsymbol{TM_{\overline{\sigma}}.K5}$
\emph{is sound wrt. the class of euclidian} $TM_{\boldsymbol{\mathcal{A}}}^{\overline{\sigma}}$
\textendash{} \emph{frames, and the logic }$\boldsymbol{TM_{\overline{\sigma}}.KT}$
\emph{is sound wrt. the class of reflexive} $TM_{\boldsymbol{\mathcal{A}}}^{\overline{\sigma}}$
\textendash{} \emph{frames.}
\end{cor}
\begin{proof}
This follows from Theorem \ref{Theorem soundness} and Lemmas \ref{Lemma axiom 4},
\ref{Lemma axiom 5}, and \ref{Lemma axiom T}.
\end{proof}
\begin{thm}
\emph{\label{Theorem completeness S4}{[}Completeness }$\boldsymbol{TM_{\overline{\sigma}}.K4}$\emph{{]}
The logic }$\boldsymbol{TM_{\overline{\sigma}}.K4}$ \emph{is complete
wrt. the class of transitive} $TM_{\boldsymbol{\mathcal{A}}}^{\overline{\sigma}}$
\textendash{} \emph{frames.}
\end{thm}
\begin{proof}
By Proposition \ref{Prop. IFF} it suffices to produce, given any\emph{
}$\boldsymbol{TM_{\overline{\sigma}}.K4}$ \textendash{} consistent
set $\Omega$, a model based on some frame from the class of all\emph{
}transitive $TM_{\boldsymbol{\mathcal{A}}}^{\overline{\sigma}}$ \textendash{}
frames, a world and a valuation satisfying $\Omega.$ Obviously, we
are going to choose the model to be $\mathcal{M}_{\Omega}^{\boldsymbol{TM_{\overline{\sigma}}.K4}}$
the world to be some $\Omega$ \textendash{} extending $\boldsymbol{TM_{\overline{\sigma}}.K4}-MCS$
$w$, and the valuation to be the canonical valuation $v^{\boldsymbol{TM_{\overline{\sigma}}.K4}}$,
such that it follows from Lemma \ref{Lemma truth} that $\mathcal{M}_{\Omega}^{\boldsymbol{TM_{\overline{\sigma}}.K4}},w\models\Omega$.
It remains to show that $\mathcal{M}_{\Omega}^{\boldsymbol{TM_{\overline{\sigma}}.K4}}$
is indeed based on a transitive frame, but cf. Lemma \ref{Lemma axiom 4}
this amounts to showing that the axiom\emph{ $\forall x\big(K_{x}\phi\rightarrow K_{x}K_{x}\phi\big)$
}is satisfied\emph{ }on $\mathcal{M}_{\Omega}^{\boldsymbol{TM_{\overline{\sigma}}.K4}}$,
which is fulfilled by Lemma \ref{Lemma axiom 4}. We conclude that
$\mathcal{M}_{\Omega}^{\boldsymbol{TM_{\overline{\sigma}}.K4}}$ is
based on a transitive $TM_{\boldsymbol{\mathcal{A}}}^{\overline{\sigma}}$
\textendash{} frame and the desired conclusion follows.
\end{proof}
The interested reader may note that compactness is a sitting duck
at this point - I will however, not make this point explicit since
I am already in excess of keystrokes. Recall our brief discussion
of Descartes cogito from the introduction; we are now equipped with
a logic in which $\forall y\exists xK_{y}(x=y)$ is a wff of the language,
and we know that it exhibits appropriate behavior as far as the relation
between syntax and semantics go. Before turning the key a few words
on Hintikka and the logic\emph{ }$\boldsymbol{TM_{\overline{\sigma}}.K4}$
are in order.

\section{Hintikka Revisited}

As an answer to Hintikka's 1962 footnote we can now present the logic
$\boldsymbol{TM_{\overline{\sigma}}.K4}$. We have, as queried, partitioned
terms in those that form well-formed sentences of the kind $K_{t}\phi$
and those which do not, and we know that the logic is sound and complete
wrt. the class of transitive $TM_{\boldsymbol{\mathcal{A}}}^{\overline{\sigma}}$
\textendash{} frames. Yet, a comment is in order. We have included
the $\boldsymbol{4}$ \textendash{} axiom $\forall x\big(K_{x}\phi\rightarrow K_{x}K_{x}\phi\big)$,
but the equivalent axiom in \cite{Hintikka1962} reads $K_{a}p\rightarrow K_{a}K_{a}p$\footnote{I will \emph{not }go into Hintikka's semantics which differs from
the kripkean semantics employed throughout the present work.}. Further, Hintikka insists that the agent whom is referred to by
$a$ \emph{must know who he is }for the axiom to make sense, and formalizes
this demand as $\exists xK_{a}(x=a)$\footnote{Again, this is formally a translation from Hintikka's framework into
the framework developed here. However, it seems rather harmless.}. As it turns out, if interpreted in the framework presented here,
there is a rather nice motivation for taking the axiom $\exists xK_{a}(x=a)$
as necessary for $K_{a}p\rightarrow K_{a}K_{a}p$ for assume we evaluate
the latter over the class of transitive frames; it turns out that
the axiom is not even valid.
\begin{prop}
\emph{\label{Prop. non-validity}{[}}$K_{a}p\rightarrow K_{a}K_{a}p$
\emph{is not valid on the class of transitive }$TM_{\boldsymbol{\mathcal{A}}}^{\overline{\sigma}}$
\textendash{} \emph{frames{]} Let $a$ be any agent-referring constant.
Then, the formula }$K_{a}p\rightarrow K_{a}K_{a}p$ \emph{is not valid
over the class of transitive} $TM_{\boldsymbol{\mathcal{A}}}^{\overline{\sigma}}$
\textendash{} \emph{frames}.
\end{prop}
\begin{proof}
We show the Proposition by constructing a counter-example. Choose
some model $\mathcal{M},$ some world $w$ in the class of transitive
$TM_{\boldsymbol{\mathcal{A}}}^{\overline{\sigma}}$ \textendash{}
frames, some wff $\phi$, a valuation $v$, and some agent-referring
constant $a$ such that $\mathcal{M},w\models_{v}K_{a}\phi$. It is
certainly possible to choose the above such that for some world $w'$
we have $(w,w')\in\mathcal{R}(a^{w,v})$. Now, since constants are
non-rigid we can even choose the interpretation such that $\mathcal{I}(a,w)\neq\mathcal{I}(a,w')$,
and thus if $t$ is yet another world such that $(w',t)\in\mathcal{R}(a^{w\text{'},v})$
we \emph{do not have by transitivity }that $(w,t)\in\mathcal{R}(a^{w,v})$.
We can even choose our model such that $\mathcal{M},t\not\models_{v}\phi$,
yielding $\mathcal{M},w'\not\models_{v}K_{a}\phi$ and subsequently
$\mathcal{\mathcal{M}},w\not\models_{v}K_{a}K_{a}\phi$.
\end{proof}
Comparing the above Proposition with the proof of Lemma \ref{Lemma axiom 4}
we see that, unsurprisingly, what goes wrong is an effect of the non-rigidity
of constants. In plain English we do not know, that the referent of
$a$ is the same in all worlds and that enables us to construct a
counter example. What would be the effect of adding Hintikka's ``Knowing
who'' \textendash{} axiom $\exists xK_{a}(x=a)$? Take world $w$,
a valuation $v$, and model $\mathcal{M}$ from the proof of Proposition
\ref{Prop. non-validity} and assume $\exists xK_{a}(x=a)$ as an
axiom. Now, as $\mathcal{M},w\models_{v}\exists xK_{a}(x=a)$ we get
by the definition of $\exists$ that $\mathcal{M},w\models_{v}\neg\forall x\neg K_{a}(x=a)$
and thus by the semantics of $\forall$ that $\mathcal{M},w\models_{v'}K_{a}(x=a)$
for some $x$ \textendash{} variant $v'$ of $v$. By the semantics
of $K_{a}$ we get for any world $w'$ with $(w,w')\in\mathcal{R}(a^{w,v'})$
that $\mathcal{M},w'\models_{v'}(x=a)$. Denote by $d$ the referent
of $x$ and note that the consequence of the above is that $\mathcal{I}(a,w')=d$
for any world $w'$ with $(w,w')\in\mathcal{R}(a^{w,v'})$. This makes
way for the last result of this paper:
\begin{prop}
\emph{\label{Prop. Knowing who}{[}Knowing Who{]} Assuming }$\exists xK_{a}(x=a)$
\emph{as an axiom yields validity of }$K_{a}p\rightarrow K_{a}K_{a}p$
\emph{for any agent-referring constant $a$ on the class of} \emph{transitive}
$TM_{\boldsymbol{\mathcal{A}}}^{\overline{\sigma}}$ \textendash{}
\emph{frames.}
\end{prop}
\begin{proof}
Take world $w$, valuation $v$, and model $\mathcal{M}$ based on
the class of transitive $TM_{\boldsymbol{\mathcal{A}}}^{\overline{\sigma}}$
\textendash{} frames and assume $\mathcal{M},w\models_{v}K_{a}p$
for agent-referring constant $a$ and wff $p$. If $w'$ is a world
such that $(w,w')\in\mathcal{R}(a^{w,v})$ we get $\mathcal{M},w'\models_{v}p$
and if $t$ is yet another world such that $(w',t)\in\mathcal{R}(a^{w',v})$
we make the observation that $\mathcal{R}(a^{w,v})=\mathcal{R}(a^{w',v})$.
By transitivity this yields $(w,t)\in\mathcal{R}(a^{w,v})$ and then
by assumption $\mathcal{M},t\models_{v}p$ which in turn implies $\mathcal{M},w'\models K_{a}p$
and $\mathcal{M},w\models K_{a}K_{a}p$ as desired.
\end{proof}
Hintikka's motivation for demanding the ``Knowing Who'' \textendash{}
axiom as prerequisite for his version of the axiom $\boldsymbol{4}$
is quiet different than the above sketched, but internal to the Kripkean
framework employed here Propositions \ref{Prop. non-validity} and
\ref{Prop. Knowing who} offers separate justification.

\section{Conclusion}

I have stated language and syntax for a two-sorted term-modal logic.
I have presented an axiomatic system for the logic $\boldsymbol{K_{TM}}^{\boldsymbol{\mathcal{A}},\overline{\sigma}}$,
and shown it to be both sound and complete wrt. the class of all $TM_{\boldsymbol{\mathcal{A}}}^{\overline{\sigma}}$
\textendash{} frames. Then I have added the ``term-modal'' version
of axiom $\boldsymbol{4}$ and shown the logic $\boldsymbol{TM_{\overline{\sigma}}.K4}$
to be both sound and complete wrt. the class of transitive $TM_{\boldsymbol{\mathcal{A}}}^{\overline{\sigma}}$
\textendash{} frames. Lastly, I have discussed the version of axiom
$4$ that Hintikka puts forward in \cite{Hintikka1962}, and suggested
a motivation for presupposing the ``Knowing Who'' \textendash{}
axiom in a Kripkean framework.

\newpage{}
\addcontentsline{toc}{section}{\refname}
%\bibliography{RendsvigBib}

\end{document}